\documentclass[11pt]{article}
\usepackage{fullpage}
\usepackage{algorithm} %
\usepackage[noend]{algpseudocode}
\usepackage{xcolor}

\def\showauthornotes{0}

\def\showdraftbox{0}

\usepackage{fancybox}

\usepackage{amsmath,amssymb,amsthm,amstext,amsfonts,bbm,algorithm,algorithmicx,graphicx,xspace,nicefrac}
\usepackage{color,stmaryrd,enumerate,latexsym,bm,amsfonts,wrapfig,verbatim,tabularx,textcomp,
subfig}
\usepackage{amsfonts}
\usepackage{comment} 
\usepackage{epsfig} 
\usepackage{latexsym,nicefrac,bbm}
\usepackage{xspace}
\usepackage{color,fancybox,graphicx,url}
\usepackage{enumitem}
\usepackage{booktabs}
\usepackage{commath}
\usepackage{mdframed}
\usepackage{pdfsync}

\usepackage{thm-restate}

\usepackage{fancybox}
\newenvironment{fminipage}%
  {\begin{Sbox}\begin{minipage}}%
  {\end{minipage}\end{Sbox}\fbox{\TheSbox}}

\newcommand{\defeq}{\stackrel{\textup{def}}{=}}

\newtheorem{theorem}{Theorem}[section]

\newtheorem{lemma}[theorem]{Lemma}
\newtheorem{definition}[theorem]{Definition}
\newtheorem{corollary}[theorem]{Corollary}

\newcommand{\diag}[1]{{\bf Diag}\left({#1}\right)}

\newcommand{\nfrac}[2]{\nicefrac{#1}{#2}} \def\abs#1{\left| #1
  \right|} \renewcommand{\norm}[1]{\ensuremath{\left\lVert #1
    \right\rVert}}

\newcommand{\pair}[1]{\left\langle{#1}\right\rangle} %

\newcommand\B{\{0,1\}}      %

\newcommand\rea{\mathbb R}

\newcommand{\V}[1]{\mathbf{#1}} %

\newcommand{\marginlabel}[1]%
{\mbox{}\marginpar{\it{\raggedleft\hspace{0pt}#1}}}

\newcommand\poly{{\textrm{poly}}}  %
 
\newcommand\polyloglog{{\textrm{polyloglog}}}

\definecolor{Mygray}{gray}{0.8}

\ifcsname ifcommentflag\endcsname\else
\expandafter\let\csname ifcommentflag\expandafter\endcsname
\csname iffalse\endcsname
\fi

\ifnum\showauthornotes=1
\newcommand{\todo}[1]{\colorbox{Mygray}{\color{red}\parbox{\textwidth}{#1}}}
\else
\newcommand{\todo}[1]{}
\fi

\ifnum\showauthornotes=1
\newcommand{\Authornote}[2]{{\sf\small\color{red}{[#1: #2]}}}
\newcommand{\Authoredit}[2]{{\sf\small\color{red}{[#1]}\color{blue}{#2}}}
\newcommand{\Authorcomment}[2]{{\sf \small\color{gray}{[#1: #2]}}}
\newcommand{\Authorfnote}[2]{\footnote{\color{red}{#1: #2}}}
\newcommand{\Authorfixme}[1]{\Authornote{#1}{\textbf{??}}}
\newcommand{\Authormarginmark}[1]{\marginpar{\textcolor{red}{\fbox{%
#1:!}}}}
\else
\newcommand{\Authornote}[2]{}
\newcommand{\Authoredit}[2]{}
\newcommand{\Authorcomment}[2]{}
\newcommand{\Authorfnote}[2]{}
\newcommand{\Authorfixme}[1]{}
\newcommand{\Authormarginmark}[1]{}
\fi

\newcommand{\paren}[1]{\left({#1}\right)}

\def\iff{\Leftrightarrow}

\newlength{\pgmtab}  %
 {
	\begin{enumerate}}{\end{enumerate}}

\def\qedsketch{\ifmmode\Box\else{\unskip\nobreak\hfil
\penalty50\hskip1em\null\nobreak\hfil$\Box$
\parfillskip=0pt\finalhyphendemerits=0\endgraf}\fi}

\newlength{\tpush}
\newcommand{\handout}[5]{
   \noindent
   \begin{center}
   \framebox{ \vbox{ \hbox to \textwidth { {\bf \coursenum\ :\  \coursename} \hfill #5 }
       \vspace{3mm}
       \hbox to \textwidth { {\Large \hfill #2  \hfill} }
       \vspace{1mm}
       \hbox to \textwidth { {\it #3 \hfill #4} }
     }
   }
   \end{center}
   \vspace*{4mm}
   \newcommand{\lecturenum}{#1}
   \addcontentsline{toc}{chapter}{Lecture #1 -- #2}
}

\ifnum\showdraftbox=1

\else

\fi

\newcommand\p{\mbox{\bf P}\xspace}

\allowdisplaybreaks

\usepackage[
    backend=biber,
    giveninits=true,
    style=alphabetic,
    url=false, 
    doi=false,
    hyperref,
    backref=true,
    backrefstyle=none,
    minnames=4,
    maxnames=10,
    sortcites
]{biblatex}

\usepackage{hyperref}

\newcommand{\sushant}{\Authornote{Sushant}}

\newcommand{\E}{\mbox{{\bf E}}}

\def\defeq{\stackrel{\mathrm{def}}{=}}
\def\setof#1{\left\{#1  \right\}}

\def\sgn#1{\mathrm{sgn} (#1)}

\def\union{\cup}

\def\abs#1{\left|#1  \right|}

\def\norm#1{\left\| #1 \right\|}

\newcommand\residual{\mathit{\alpha}}

\newcommand\FasterGammaApprox{Gamma-Solver}

\newcommand{\etal}{\emph{et al}}

\def\aa{\pmb{\mathit{a}}}
\newcommand\bb{\boldsymbol{\mathit{b}}}
\newcommand\cc{\boldsymbol{\mathit{c}}}
\newcommand\dd{\boldsymbol{\mathit{d}}}

\newcommand\ff{\boldsymbol{\mathit{f}}}
\renewcommand\gg{\boldsymbol{\mathit{g}}}

\newcommand\rr{\boldsymbol{\mathit{r}}}
\renewcommand\ss{\boldsymbol{\mathit{s}}}
\def\tt{\boldsymbol{\mathit{t}}}

\newcommand\vv{\boldsymbol{\mathit{v}}}
\newcommand\ww{\boldsymbol{\mathit{w}}}
\newcommand\yy{\boldsymbol{\mathit{y}}}
\newcommand\zz{\boldsymbol{\mathit{z}}}
\newcommand\xx{\boldsymbol{\mathit{x}}}

\renewcommand\AA{\boldsymbol{\mathit{A}}}
\newcommand\BB{\boldsymbol{\mathit{B}}}
\newcommand\CC{\boldsymbol{\mathit{C}}}

\newcommand\MM{\boldsymbol{\mathit{M}}}
\newcommand\LL{\boldsymbol{\mathit{L}}}
\newcommand\RR{\boldsymbol{\mathit{R}}}

\newcommand\UU{\boldsymbol{\mathit{U}}}

\newcommand\VV{\boldsymbol{\mathit{V}}}

\newcommand\ZZ{\boldsymbol{\mathit{Z}}}

\newcommand\AAhat{\boldsymbol{\widehat{\mathit{A}}}}

\newcommand\ZZhat{\boldsymbol{\widehat{\mathit{Z}}}}

\newcommand\CChat{\boldsymbol{\widehat{\mathit{C}}}}

\newcommand\xxtil{\boldsymbol{\widetilde{\mathit{x}}}}
\newcommand\yytil{\boldsymbol{\widetilde{\mathit{y}}}}
\newcommand\yyhat{\boldsymbol{\widehat{\mathit{y}}}}

\newcommand\Otil{\widetilde{O}}

\newcommand\ddhat{\boldsymbol{\widehat{d}}}
\newcommand\rrhat{\boldsymbol{\widehat{r}}}
\newcommand\rrtil{\boldsymbol{\widetilde{r}}}
\newcommand\xxhat{\boldsymbol{\widehat{x}}}

\newcommand{\one}{\mathbf{1}}

\DeclareMathOperator*{\argmin}{arg\,min}

\newenvironment{tight_enumerate}{
\begin{enumerate}
 \setlength{\itemsep}{2pt}
 \setlength{\parskip}{1pt}
}{\end{enumerate}}
\newenvironment{tight_itemize}{
\begin{itemize}
 \setlength{\itemsep}{2pt}
 \setlength{\parskip}{1pt}
}{\end{itemize}}

\newcommand{\gammap}{\gamma_{p}}
\newcommand{\pnorm}[1]{\norm{#1}_{p}}
\newcommand{\opt}{\textsc{OPT}}
\newcommand{\eps}{\varepsilon}

\newcommand{\smallnorm}[1]{\|#1\|}

\newcommand{\energy}[1]{\Psi\left(#1\right)}

\begin{document}

\title{Iterative Refinement for $\ell_p$-norm Regression
\footnote{This paper has been published at SODA 2019~\cite{AdilKPS19}, and was
  initially submitted to SODA on July 12, 2018.}}
\author{
  Deeksha Adil\thanks{University of Toronto. 
    \texttt{deeksha@cs.toronto.edu}.
    Supported by an Ontario Graduate Scholarship, and by a Connaught
    New Researcher award to Sushant Sachdeva.
  }
  \and
  Rasmus Kyng\thanks{  Harvard.   \texttt{rjkyng@gmail.com}.
    Supported by ONR grant
    N00014-18-1-2562.}
  \and
  Richard Peng\thanks{
  Georgia Tech.
  \texttt{richard.peng@gmail.com}. Supported in part by the National Science Foundation under Grant No. 1718533.}
  \and
  Sushant Sachdeva  \thanks{University of Toronto. \texttt{sachdeva@cs.toronto.edu}. Research supported in part by the Natural Sciences and
    Engineering Research Council of Canada (NSERC), and a Connaught
    New Researcher award.}  }

\maketitle
\thispagestyle{empty}

\begin{abstract}

  We give improved algorithms for the $\ell_{p}$-regression problem,
  $\min_{\xx} \|\xx\|_{p}$ such that $\AA\xx=\bb,$ for all
  $p \in (1,2) \cup (2,\infty).$ Our algorithms obtain a high accuracy
  solution in
  $\Otil_{p}(m^{\frac{|p-2|}{2p + |p-2|}}) \le
  \Otil_{p}(m^{\nfrac{1}{3}})$ iterations, where each iteration
  requires solving an $m \times m$ linear system, with $m$ being the
  dimension of the ambient space.
  
  Incorporating a procedure for maintaining an approximate inverse of
  the linear systems that we need to solve at each iteration, we give
  algorithms for solving $\ell_{p}$-regression to $1 / \poly(n)$
  accuracy that runs in time $\Otil_p(m^{\max\{\omega, 7/3\}}),$ where
  $\omega$ is the matrix multiplication constant.  For the current
  best value of $\omega > 2.37$, this means that we can solve
  $\ell_{p}$ regression as fast as $\ell_{2}$ regression, for all
  constant $p$ bounded away from $1.$

  Our algorithms can be combined with nearly-linear time solvers for
  linear systems in graph Laplacians to give minimum $\ell_{p}$-norm
  flow / voltage solutions to $1 / \poly(n)$ accuracy on an undirected
  graph with $m$ edges in
  $\Otil_{p}(m^{1 + \frac{|p-2|}{2p + |p-2|}}) \le
  \Otil_{p}(m^{\nfrac{4}{3}})$ time.

  For sparse graphs and for matrices with similar dimensions, our
  iteration counts and running times improve upon the $p$-norm
  regression algorithm by [Bubeck-Cohen-Lee-Li STOC`18], as well as
  general purpose convex optimization algorithms.  At the core of our
  algorithms is an iterative refinement scheme for $\ell_{p}$-norms,
  using the quadratically-smoothed $\ell_{p}$-norms introduced in the
  work of Bubeck \etal.  Formally, given an initial solution, we
  construct a problem that seeks to minimize a quadratically-smoothed
  $\ell_{p}$ norm over a subspace, such that a crude solution to this
  problem allows us to improve the initial solution by a constant
  factor, leading to algorithms with fast convergence.

\end{abstract}

\newpage
\tableofcontents

\newpage
\setcounter{page}{1}

\section{Introduction}
Iterative methods that converge rapidly to a solution are of
fundamental importance to numerical analysis, optimization, and more
recently, graph algorithms.  In the study of iterative methods, there
are significant discrepancies between iterative methods geared towards
linear problems, and ones that can handle more general convex
objectives.  For systems of linear equations, which corresponds to
minimizing $\ell_{2}$-norm objectives over a subspace,
most
iterative methods obtain $\epsilon$-approximate solutions in iteration
counts that scale as $\log(1 / \epsilon)$.  More generally, for
appropriately defined notions of accuracy, a constant-accuracy
linear system solver can be iterated to give a much higher accuracy
solver using a few  calls to the crude solver.  Such phenomena are not limited to
linear systems either: an algorithm that produces approximate maximum
flows on directed graphs can be iterated on the residual graph to
quickly obtain high-accuracy answers.

On the other hand, for the much wider space of non-linear optimization
problems arising from optimization and machine learning, it's
significantly more expensive to obtain high accuracy solutions.  Many
widely used methods such as (accelerated) gradient descent, obtain
$\epsilon$-approximate answers using iteration counts that scale as
$\poly(1 / \epsilon).$ Such discrepancies also occur in the overall
asymptotic running times.  An important and canonical problem in this
space is $\ell_p$-norm regression:
\begin{align}
  \label{eq:primal}
  \tag{*}
  \min_{\xx \in \rea^{m}: \AA \xx = \bb} \pnorm{\xx}^{p}, 
\end{align}
for some $\AA \in \rea^{n \times m} (m \ge n),$ and
$\bb \in \rea^{n}.$
For $p=2,$ this corresponds exactly to solving a linear system, and
hence is solvable by a matrix inversion in $O(m^\omega)$ time
\footnote{$\omega$ is the matrix multiplication exponent. Currently we
  know $\omega \le 2.3728639..$~\cite{Williams12,Legall14}.}
For $p = 1$ and $p = \infty,$ this problem is inter-reducible
to linear programming~\cite{Tillmann13:book,Tillmann15,BubeckCLL18}.

Interior point methods also allow us to solve $\ell_{p}$-norm
regression problems in $\sqrt{rank}$ iterations~\cite{NesterovN94, LeeS14}, where each iteration
requires solving an $m \times m$ linear system for any
$p \in [1,\infty]$.  Bubeck \etal~\cite{BubeckCLL18} show that this
iteration count is tight for the interior point method framework, and
instead propose a different method which requires only
$\Otil_p(m^{\abs{\frac{1}{2}-\frac{1}{p}}})$ \footnote{$O_p(\cdot)$
  notation hides constant factors that depend on $p$, and its dual
  norm $\frac{p}{p - 1}.$ $\Otil_p(\cdot)$ notation also hides
  $\poly(\log \frac{mn}{\eps})$ factors in addition.} iterations for
$p \neq 1,\infty,$ which for large constant $p$ still tends to about
$m^{1/2}$.  On the other hand, $\epsilon$-approximate solutions can be
computed in about $m^{1/3} \poly(1 / \epsilon)$
iterations~\cite{ChinMMP13} \footnote{this result only addressed the
  $p = \infty$ case, but its techniques generalize to all other $p$}.

Furthermore, this discrepancy also carries over to the graph theoretic
case.  If the matrix $\AA$ is the vertex-edge incidence matrix of a
graph, then this problem captures graph problems such as $p$-norm
Lipschitz learning and finding $\ell_{p}$-norm minimizing flows
meeting demands given by $\bb.$ Here low accuracy approximate solutions can be
obtained in nearly-linear time when $p=\infty$~\cite{Peng16,
  Sherman17b}, and almost-linear time for all other values of
$p$~\cite{Sherman17a,Sidford17:talk}.
However, the current best high accuracy solutions take at least
$m^{\min\{10/7, 1+ \abs{1/2 - 1/p}\}}$ time~\cite{Madry13,
  BubeckCLL18}.

\subsection{Contributions}

\paragraph{Iterative Refinement for $\ell_{p}$-norms.}
In this paper, we propose a new iterative method for $\ell_{p}$-norm
regression problems~\eqref{eq:primal} that achieves geometric
convergence to the optimal solution. Our method only requires solving
$O_{p}(\log \nfrac{1}{\eps})$ \emph{residual problems} to find an
$\eps$-approximate solution, or $O_{p}(\kappa\log \nfrac{1}{\eps})$
residual problems, each solved to a $\kappa$-approximation
factor.
Such an iterative method was previously known only for $p=2$
and $\infty.$ Curiously, our residual problems look very similar to
the original problem~\eqref{eq:primal}, with the $\ell_{p}$ norms
replaced by their quadratically-smoothed versions introduced by
Bubeck~\etal~\cite{BubeckCLL18}.
This result, Theorem~\ref{thm:IterativeRefinement}, can be stated
informally as:
\begin{theorem}
\label{thm:IterativeRefinementInformal}
There exists a class of residual problems for $p$-norm regression
(which we will define in Definition~\ref{def:ResidualProblem})
such that any $p$-norm regression problem can be solved to
$\epsilon$-relative accuracy by solving to relative error $\kappa$
a sequence of
$O_p ( \kappa \log(\frac{m}{\eps}))$ residual problems.
\end{theorem}

\paragraph{Improved Iteration Count for $\ell_{p}$-Regression.}
We then give an algorithm for quickly solving the residual problem
motivated by the approximate maximum flow by electrical flows algorithm
by Christiano et al.~\cite{ChristianoKMST10} and its generalizations
to regression problems~\cite{ChinMMP13}.
This is given as Theorem~\ref{thm:MainResult}, and can be stated
informally as:
\begin{theorem}
\label{thm:MainResultInformal}
For any $p > 2$,
an instance of a residual problem for $p$-norm regression
as defined in Definition~\ref{def:ResidualProblem} can be solved
in $\Otil_{p}(m^{\frac{p-2}{3p-2}})$ iterations, each of which
consist of solving a system of linear equations plus updates
that take linear time.
\end{theorem}
This improves on the work of Bubeck~\etal~\cite{BubeckCLL18}
for all $p > 2,$ with
the number of iterations equaling $\Otil_{p}(1)$ for $p=2$
(essentially the same as Bubeck~\etal) and
tending to $\Otil(m^{\nfrac{1}{3}})$ as $p$ goes to $\infty$ (compared
to $\Otil(m^{\nfrac{1}{2}})$ for Bubeck~\etal).
However, our results don't give anything for $p=\infty$ 
due to the dependency in $p$ in the $\Otil_{p}(\cdot)$ term.
It's worth noting that even in the constant error regime,
this improves by a factor of about $\min\{m^{\frac{(p-2)^2}{2p(3p - 2)}},
m^{\frac{4}{3(3p - 2)}}\}$ over the current state of the art,
which for small $p$ is due to Bubeck et al.~\cite{BubeckCLL18},
and for large $p$ is based on unpublished modifications to
Christiano et al.~\cite{ChristianoKMST10,Madry11:communication}.

\paragraph{A Duality Based Approach to $\ell_{p}$-Regression.}
For the remaining case of $1 < p < 2$, we instead solve the dual problem,
which is a $\frac{p}{p - 1}$-norm regression problem, and utilize its
solution to solve our original $\ell_{p}$-Regression.  This leads to
iteration counts of the form
$\Otil_{p}(m^{\frac{2 - p}{p + 2}} \log(1 / \epsilon))$ for solving such
problems.  Note that this result also does not give anything when
$p = 1$, as the constants related to its dual norm, $\frac{p}{1 - p}$
become prohibitive. For all $p \in (1,\infty),$ our
iteration count achieves the following exponent on $m,$
\[
\frac{\abs{\frac{1}{2} - \frac{1}{p}}}{1 + \abs{\frac{1}{2} - \frac{1}{p}}},
\]
while the exponent from the previous result~\cite{BubeckCLL18} is
$\abs{\frac{1}{2} - \frac{1}{p}}$: our algorithm has better dependence
on $m$ on all constant $p$ (albeit with larger constants depending on
$p$).

For the case of $p = 4$, a manuscript by
Bullins~\cite{Bullins18:arxiv} from December 2018 (after our paper was
accepted to SODA 2019, but independently developed), gives the same
iteration count as our algorithm of $n^{1/5}\log(1/\epsilon)$ up to
polylogs. Bullins' approach requires a linear system solve per
iteration, similar to our approach when implemented without inverse
maintenance.  Bullins' algorithm is based on higher-order
acceleration, and the agreement between running times suggests there
may be a strong connection between our ``accelerated'' multiplicative
weight method and his accelerated gradient-based method.

\paragraph{Faster $\ell_{p}$-Regression.}
Our improved iteration counts can be readily combined with methods
for speeding up optimization algorithms that utilize linear system solvers,
including inverse maintenance~\cite{Vaidya89,LeeS15}.
This results in an $\Otil_{p}(m^{\max\{\omega, \nfrac{7}{3}\}} )$ time
algorithm for solving $\ell_{p}$ regression problems for all $p \in (1,\infty)$,
which we formalize in Theorem~\ref{thm:FasterMatrixAlgo}.

This bound for $p$-norm regression with general matrices brings us to the somewhat
surprising conclusion that for the current value of $\omega > 7/3$,
$p$-norm regression problems
(with constant $p$ that's also constant-bounded away from $1$)
on square matrices can be solved as fast as solving the underlying linear
systems, or equivalently, $\ell_2$ regression problems.

This is based on maintaining an approximate inverse to the linear
systems we need to solve in each step of the iterative method as
pioneered by Vaidya~\cite{Vaidya89}.  However, our modification
interacts directly with the potential functions we use to control
iteration counts in the inner loop of our iterative method.  A similar
approach for maintaining an approximate inverse was used by Cohen et
al.~\cite{CohenLS18:arxiv} to give an $\tilde{O}(m^\omega)$ algorithm for Linear Programming,
after our initial submission to SODA, but before our paper was
publicly available.  Both works build on ideas developed by Cohen,
see~\cite{Lee17:talk}.

\paragraph{Faster $p$-Norm Flows.}

When solving $p$-norm flow problems,
our algorithm can made faster by using  
Laplacian solvers for graph
problems~\cite{Vaidya90,Teng10:survey} to solve the linear equations
that arise during our iterations.
This gives algorithms for finding $p$-norm flows on undirected graphs
to accuracy $\epsilon$ with running time
$\Otil_{p} \left(m^{1 + \frac{\abs{p - 2}}{2p + \abs{p - 2}}}
\log(1 / \epsilon)\right)$ for $p \in (1, \infty)$
via direct invocations of fast Laplacian
solvers~\cite{SpielmanTengSolver:journal}.

Our results thus give the first evidence that wide classes of graph optimization
problems can be solved in time $m^{4/3}$ or faster.
While such a bound (via. fast Laplacian solvers) is by now well-known
in the approximate setting~\cite{ChristianoKMST10}, the $m^{10/7}$
iteration bounds due to Madry~\cite{Madry13,Madry16} represent the only
results to date in this direction for high accuracy answers on sparse
graphs.

\paragraph{Generalizations and Extensions.}
  While we focus on Problem~\eqref{eq:primal},
  under mild assumptions about polynomially bounded objectives,
  we can solve the following more general problem:
\begin{align*}
\min_{\xx} \qquad & \norm{\CC \xx - \dd}_{p}\\
& \AA \xx = \bb
\end{align*}
The reduction is discussed in Section~\ref{sec:Variants}.
The combination of
an affine constraint on $\xx$ with an affine transformation in the
$p$-norm objective means we can solve most variants of $p$-norm
optimization problems.

Similar ideas can be used to solve $p$-norm Lipschitz learning
problems~\cite{KyngRSS15} on graphs quickly.

\subsection{Comparison to Previous Works}

\paragraph{Numerical Methods and Preconditioning.}

Iterative methods and preconditioning are the most fundamental
tools in numerical algorithms~\cite{Axelsson94:book,Saad03:book}.
As studies of such methods often focus on linear problems, many existing
analyses of iterative methods are restricted to linear systems.
Generalizing such methods, as well as numerical methods,
to broader settings is a major topic of
study~\cite{Henson03,NocedalW06:book,Kelley99:book,KnollK04}.

The study of more efficient algorithms for combinatorial flow problems
has benefited enormously from ideas from linear and non-linear preconditioning.
Recent advances in approximate maximum flow and transshipment
algorithms~\cite{LeeRS13,Sherman13,KelnerLOS14,
  RackeST14,GhaffariKKLP15,Peng16,Sherman17a,Sherman17b} build upon
such ideas.  However, these methods rely on the preconditioner being a
linear operator, and give $\poly(1 / \epsilon)$ dependence. 

\paragraph{Optimization Algorithms.}
Our techniques for solving the residual problems are directly
motivated by approximating maximum flow using electrical
flows~\cite{ChristianoKMST10}.  While this algorithm has been extended
to multicommodity flows and regression
problems~\cite{KelnerMP12,ChinMMP13}, all these results have
$\poly(1 / \epsilon)$ dependencies.

Several recent results for obtaining $\log(1 / \epsilon)$ dependencies are
all motivated by convex optimization techniques.  In particular, the
state of the art running times are by interior point methods.  These
include directly modifying the interior point method
(IPM)~\cite{LeeS14,LeeS15,KyngRSS15},
combining techniques from the electrical flow algorithms with IPM update
steps~\cite{Madry13,KyngRSS15,KyngRS15,Madry16,CohenMSV17}, and
increasing the `confidence interval', and in turn step lengths, of the
IPM update steps~\cite{ALOW17,CohenMTV17,BubeckCLL18}.
Our result based on creating intermediate problems has the most
in common with the last of these.
However, our method differs in that our guarantees
for this intermediate problem holds over the entire space.

\paragraph{Inverse Maintenance}

Our final running time of
$\Otil_{p}(m^{\max\{\omega, \nfrac{7}{3}\}} )$ for
$\ell_{p}$-regression incorporates inverse maintenance.  This is a
method introduced by Vaidya~\cite{Vaidya89} for speeding up
optimization algorithms for solving minimum cost and multicommodity
flows.  It takes advantage of the controllable rate at which
optimization algorithms modify the solution variables to reuse
inverses of matrices constructed from such variables.

Previous studies of inverse maintenance~\cite{Vaidya89,LeeS14,LeeS15}
have been geared towards the interior point method.
Here the norm per update step can be controlled, and we believe this also
holds for their applications in faster cutting plane methods~\cite{LeeSW15}.
While such methods also give gains in the case of our algorithm,
for the final bound of about $m^{\omega}$, we instead bound the
progress of the steps against a global potential function motivated
by the electrical flow max-flow algorithm~\cite{ChristianoKMST10}.

\paragraph{Speedups for Matrices with Uneven Dimensions}

Our algorithm on the other hand does not take into account
sparsity of the input matrix, or possibly uneven dimensions
(e.g. $m \approx n^{2})$).
In these settings, the methods based on accelerated stochastic gradient
descent from~\cite{BubeckCLL18} obtain better performances. 
On the other hand, we believe our methods have the potential of extending
to such settings by combining the intermediate problems with
$\ell_p$ row sampling~\cite{CohenP15}.
However, analysis of such row sampling routines for our residual problems
containing mixed $\ell_2$ and $\ell_p$ norm functions is outside the
scope of this paper.

\section{Technical Overview}

\paragraph{Iterative Refinement for $\ell_{p}$-norms.}  To design
their algorithm for $\ell_{p}$-norm regression,
Bubeck~\etal~\cite{BubeckCLL18} construct a function $\gammap(t,x),$
which is $C_1$, \footnote{A function is said to be $C_1$ if
  it's continuous, differentiable, and has a continuous derivative},
 quadratic in the range $\abs{x} \le t,$ and behaves as
$|x|^{p}$ asymptotically (see Def.~\ref{def:gamma}). Our key lemma
states one can locally approximate $\pnorm{x + \Delta}^{p}$ as a linear
function plus a $\gammap(|x|, \Delta)$ ``error'' term \footnote{It is
  useful to compare the $\gammap$ term to the second-order Hessian
  term in Taylor expansion}
(Lemma~\ref{lem:LocalApprox}):
\[|x + \Delta|^{p} = \abs{x}^{p} + \Delta \frac{d}{dx}\abs{x}^{p} +
  O_{p}(1)\gammap(\abs{x},\Delta).\] Surprisingly, this approximation
only has an $O_p(1)$ ``condition number''. Proceeding just as for
gradient descent, or Newton's method, means that if at each step we
solve the following local approximation problem to a factor $\kappa,$
\[\max_{\AA \Delta = 0} \gg^{\top}\Delta -  \gammap(\xx,
  \Delta),\] where $\gg$ is the gradient of our loss function, we can
converge to an $\eps$-approximate solution in roughly
$\kappa \log \nfrac{1}{\eps}$ iterations
(Theorem~\ref{thm:IterativeRefinement}).

\paragraph{Improved Algorithms for $\ell_{p}$-regression for $p \ge 2$.}
The key advantage afforded by our iterative algorithm is that we now
only need to design a algorithm for the residual problem that achieves
a crude approximation factor (we achieve $O_p(1)$). As a first step,
by a binary search and some rescaling, we show (Lemma~\ref{lem:Dual})
that it suffices to achieve a constant factor approximation to
$O_p(\log \nfrac{m}{\eps})$ problems of the following form,
\[
  \min_{\AA\xx = 0, \gg^{\top}\xx=c} \gammap(\tt, \xx).
\]
The technical heart of our proof is to give an algorithm
(\FasterGammaApprox, Algorithm~\ref{alg:FasterOracleAlgorithm})
inspired by the multiplicative weight update (MWU) method
(see \cite{MultiplicativeSurvey12} for a survey), combined with the
width-reduction inspired by the faster flow algorithm of
Christiano~\etal.~\cite{ChristianoKMST10}, and its matrix version by
Chin~\etal.\cite{ChinMMP13}.  At each iteration, we solve a weighted
$\ell_{2}$ minimization problem to find the next update step. If this
update step has small $\ell_{p}$ norm, we add this to our current
solution, and update the weights. Otherwise, we identify a set of
coordinates that have small current weights, and yet are contributing
most of the $\ell_{p}$ norm, and we penalize them by increasing their
weights (and do not add our update step to the current
solution). Setting the parameters carefully, we show that after
$\Otil_{p}(m^{\frac{p-2}{3p-2}})$ iterations, the average of the
update steps achieves an $O_{p}(1)$-approximation to our modified
residual problem (Theorem~\ref{lem:FasterAlgo}). Combining this with
our iterative refinement algorithm, we obtain our algorithms for
$\ell_{p}$-norm regression that require only
$\Otil_{p}(m^{\frac{p-2}{3p-2}})$ iterations (or linear system
solves).

\paragraph{Maintaining Inverses for Improved Algorithm.}

Our inverse maintenance procedure utilizes the same combination of
low-rank updates and matrix multiplications as in previous
results~\cite{Vaidya89,LeeS14,LeeS15}.
However, the rate of convergence of our algorithm, and in turn the rate
at which we adjust the weights from the MWU procedure, are governed
by growths in the $\ell_2$ minimization problem.
This leads to the difficulty of uneven progress across the iterations.

We solve this issue by a simple yet subtle scheme motivated
by lazy updates in data structures~\cite{GuptaP13,AbrahamDKKP16}.
We bucket changes to the values of entries based on their
magnitudes, and update entries that received too many updates
of a certain magnitude separately.
This differs with previous methods that update weights exceeding
approximation thresholds as they happen,
and enables a closer interaction with the overall potential function
based convergence analysis.

\section{Preliminaries}
\label{sec:Prelims}

We use the following family of functions, $\gammap(t,x)$ defined in \cite{BubeckCLL18}. 
\begin{definition}[$\gamma_p$ function]
  \label{def:gamma}
For $t \geq 0$ and $p \ge 1$, define
\begin{equation*}
  \gammap(t, x) =
  \begin{cases}
    \frac{p}{2}t^{p-2}x^2 & \text{if } |x| \leq t,\\
    |x|^p + (\frac{p}{2} - 1)t^p & \text{otherwise }.
  \end{cases}
\end{equation*}
\end{definition}
\noindent These functions can be thought of a quadratic approximation of $\abs{x}^p$ in a small range around zero. The following properties follow directly from the definition.
 \begin{tight_enumerate}
\item $\gammap(0,x) = \abs{x}^p$.
\item $\gammap(t,x)$ is quadratic in the range $-t \leq x \leq t$.
\item $\gammap$ is $C^{1}$ in both $x, t.$
\end{tight_enumerate}
\noindent We show several other important properties of $\gammap$ in
the following lemmas. Their proofs are straightforward and deferred to
Appendix \ref{sec:ProofsSec3}

\begin{restatable}{lemma}{Gamma}
\label{lem:gamma}
Function $\gammap$ is as defined above.
  \begin{tight_enumerate}
  \item For any $p \ge 2,$ $t \ge 0,$ and $x\in \mathbb{R},$ we have
    $\gammap(t,x) \ge |x|^{p},$ and
    $\gammap(t, x) \ge \frac{p}{2} t^{p-2}x^{2}.$
  \item It is homogeneous under rescaling of \emph{both} $t$ and
    $x,$ i.e., for any $t, \lambda \ge 0, p \geq 1,$ and any $x$ we have
    $ \gammap\left(\lambda t, \lambda x\right) = \lambda^{p}
    \gammap\left(t, x\right).$
  \item For any $t > 0, p\geq1$ and any $x,$ we have
    $\gammap^{\prime}(t, x) = p\max\{t, \abs{x}\}^{p-2}x.$
  \end{tight_enumerate}
\end{restatable}
\noindent The next lemma shows a bound on the value of $\gammap$ when $x$ is scaled up or down. %

\begin{restatable}{lemma}{Rescaling}
\label{lem:Rescaling}
  For any $p>1, \Delta \in \mathbb{R}$ and $\lambda \geq 0$, we have,
  \[ {\min\{2,p\}} \leq x\frac{\gammap'(t,x)}{\gammap(t,x)} \leq \max \{2,p\}.\]
  This implies, 
  \begin{equation*}
 \min\{\lambda^2,\lambda^p\} \gammap(t,\Delta)  \leq  \gammap(t,\lambda \Delta) \leq \max\{\lambda^2,\lambda^p\} \gammap(t,\Delta).
  \end{equation*}
  \end{restatable}
  \noindent
  The following lemma allows us to bound the second order change in
  $\gammap(\xx)$ as $\xx$ changes to $\xx + \Delta.$
\begin{restatable}{lemma}{FirstOrderAdditive}
  \label{lem:FirstOrderAdditive}
  For any $p\geq 2, t \ge 0$ and any $x, \Delta,$
  we have
\[
 \gammap(t,x+\Delta) \le \gammap(t,x) + \abs{\gammap'(t,
      x)\Delta} \\ + p^{2} 2^{p-3} \max\{t, \abs{x}, \abs{\Delta}\}^{p-2}
    \Delta^{2}.
\]
\end{restatable}

\paragraph{Notation.}
For a vector $\xx,$ let $\abs{\xx}$ denote the vector with its
$i^{\text{th}}$ coordinate as $\abs{\xx_i}.$ For any two vectors $\tt$
and $\xx$, $\gammap(\tt,\xx)$ denotes the sum
$\sum_i \gammap(\tt_i,\xx_i)$.

\section{Main Iterative Algorithm}
\label{sec:MainAlgo}

In this section we analyze procedure {\sc p-Norm}, i.e., Algorithm \ref{alg:MetaAlgo}. Our main result for this section is,
\begin{theorem}[$\ell_{p}$-norm Iterative Refinement]
  \label{thm:IterativeRefinement}
  For any $p  \in (1,\infty),$ and $\kappa \ge 1.$ Given an initial feasible
  solution $\xx^{(0)}$ (Definition \ref{start}) to our optimization problem
  (Equation~\eqref{eq:primal}), Algorithm \ref{alg:MetaAlgo} finds an $\epsilon$-approximate
  solution to \eqref{eq:primal} in
  $O_p \left( \kappa \log\left(\frac{m}{\eps}\right)\right)$ calls
  to a $\kappa$-approximate solver for the residual problem (Equation \eqref{eq:residual}).
\end{theorem}

\begin{algorithm}
\caption{Meta-Iterative Algorithm}
\label{alg:MetaAlgo}
 \begin{algorithmic}[1]
 \Procedure{p-Norm}{$\AA, \bb, \eps$}
 \State $\xx^{(0)} \leftarrow  \min_{\AA\xx=\bb} \norm{\xx}^2_2.$
 \State $T \leftarrow O_p(\kappa \log \left(\frac{m}{\eps}\right))$
 \State $\lambda \leftarrow \left(\frac{p-1}{p4^p}\right)^{\frac{1}{\min (1,p-1) }}$
 \For{$t = 0$ to $T$} 
\State{$\Delta \leftarrow$ {\sc $\kappa$-Approx}$(\AA,\bb,\xx^{(t)}),\lambda,\xx^{(0)}$}\quad \Comment{$\kappa$-approximate solution to \eqref{eq:residual}}
 \State $\xx^{(t+1)} \leftarrow \xx^{(t)} - \lambda \Delta$
 \If{$\smallnorm{\xx^{(t+1)}}_p^p \geq \smallnorm{\xx^{(t)}}_p^p$}
\State \Return $\xx^{(t)}$
\EndIf
\EndFor
\State \Return $\xx^{(T)}$
 \EndProcedure 
 \end{algorithmic}
\end{algorithm}

\noindent The theorem says that it is sufficient to solve an instance
of the residual problem \eqref{eq:residual} crudely, and only a
logarithmic number of times. Before we prove the theorem, we define
the terms used in the statement and prove some results that would be
needed for the proof. We begin by defining an $\eps$-approximate
solution to our main optimization problem.
\begin{definition}[$\eps$-approximate solution] \label{def:epsApprox}
We say our solution $\xx$ is an $\eps$-approximate solution to \eqref{eq:primal} if $\AA\xx =\bb$ and 
\[
\smallnorm{\xx}_p^p \leq (1 + \eps) \smallnorm{\xx^{\star}}_p^p,
\]
where $\xx^{\star}$ is the $\opt$ of \eqref{eq:primal}.
\end{definition}
\noindent We next define what we use as our residual problem and what we mean by a $\kappa$-approximate solution.

\begin{definition}[Residual Problem]\label{def:ResidualProblem}
For any given $\xx$ and $p>1$, let
\[\residual(\Delta) \defeq \langle \gg, \Delta \rangle - 
 \frac{p-1}{p2^p} \gammap(|\xx|,\Delta),\] where $\gg$ is the gradient,
$\gg = p\abs{\xx}^{p-2} \xx$.  We call the following problem to
be the {\em residual problem} of \eqref{eq:primal} at $\xx$.
 \begin{equation}
   \label{eq:residual}
\max_{\AA \Delta = 0} \residual(\Delta).
\end{equation}
\end{definition}

\begin{definition}[$\kappa$-approximate solution]
\label{ApproxSolver}
Let $\kappa \geq 1$. A $\kappa$-approximate solution for the residual problem is
$\tilde{\Delta}$ such that $\AA \tilde{\Delta} = 0$ and,
$
\residual(\tilde{\Delta}) \geq \frac{1}{\kappa}
\residual(\Delta^{\star}).
$
Here $\Delta^{\star} = \max_{\AA \Delta = 0} \residual(\Delta)$.
\end{definition}

\noindent In order to see why we choose this problem as our residual problem we show that the objective of the residual problem bounds the change in $p$-norm of a vector $\xx$ when perturbed by $\Delta$ (Lemma \ref{lem:LambdaBound}).
\begin{restatable}{lemma}{BoundInitial}
    \label{lem:LocalApprox}
    Let $p \in (1,\infty)$. Then for any $\xx$ and any $\Delta$,
\[
\abs{x}^p + g \Delta + \frac{p-1}{p2^p} \gammap(\abs{x},\Delta) \leq \abs{x + \Delta}^p \\ \leq  \abs{x}^p + g \Delta + 2^p\gammap(\abs{x},\Delta),
\]
where $g = p\abs{x}^{p-2}x$ is the derivative of the function
$\abs{x}^p$.
\end{restatable}
The proof can be found in Appendix \ref{sec:ProofsSec4}.  

\begin{lemma} \label{lem:LambdaBound} Let $p  \in (1,\infty)$ and $\lambda$ be such that
  $\lambda^{\min\{1,p-1\}} \le \frac{p-1}{p4^p}$. Then for any $\Delta$ we
  have,
\[
  \norm{\xx}_{p}^{p} - \residual(\lambda \Delta) \leq \norm{\xx -\lambda \Delta}_{p}^{p}
  \leq \norm{\xx}_{p}^{p} - \lambda \residual(\Delta).
\]
\end{lemma}
\begin{proof}
Applying lemma \ref{lem:LocalApprox} to all the coordinates, we obtain,
\[\label{precondition}
  \pnorm{\xx}^p - \langle g, \Delta \rangle + \frac{p-1}{p2^p}  \gammap(|\xx|,\Delta) \leq \pnorm{\xx - \Delta}^p\\  \leq  \pnorm{\xx}^p - \langle g, \Delta \rangle + 2^p \gammap(|\xx|,\Delta).
\]
Using definition \ref{def:ResidualProblem}, equation \eqref{precondition} directly implies,
$\norm{\xx - \Delta}_{p}^{p} \ge \norm{\xx}_{p}^{p} - \residual(\Delta)$ for
all $\Delta.$ Now to prove the other side, note that for any $\lambda \in [0,1],$ and any $\Delta,$ we have from Lemma 
\ref{lem:LocalApprox} and Lemma \ref{lem:gamma}
\[
   \norm{\xx - \lambda \Delta}_{p}^{p}
     \le \norm{\xx}_{p}^{p}
      - \pair{g, \lambda \Delta} + 2^p
      \gammap(|\xx|,\lambda \Delta) 
  \\
       \le 
      \norm{\xx}_{p}^{p}
      - \lambda \left( \langle g, \Delta \rangle -
      \lambda^{\min\{1,p-1\}} 2^p
      \gammap(|\xx|,\Delta) \right). 
\]
Picking $\lambda$ such that $\lambda^{\min\{1,p-1\}} \le \frac{p-1}{p4^p},$
we obtain that for any $\Delta,$
\[
  \label{delta3}
    \norm{\xx - \lambda \Delta}_{p}^{p}
    \le
    \norm{\xx}_{p}^{p}
    - \lambda \left( \langle g, \Delta \rangle - \frac{p-1}{p2^p} 
    \gammap(|\xx|,\Delta) \right) \\
    =
    \norm{\xx}_{p}^{p} -  \lambda \residual(\Delta),
\]
thus concluding the proof of the lemma.
\end{proof}

 \noindent For any iterative algorithm we need a starting feasible solution. We could potentially start with any feasible solution but we define the following starting solution which we claim is a {\it good} starting point. Lemma \ref{lem:initialpoint} shows us that our chosen starting point is only polynomially away from the optimum solution, and is thus a {\it good} choice. The proof of the lemma can be found in Appendix \ref{sec:ProofsSec4}.
\begin{definition}[Initial Solution]
\label{start}
We define $\xx^{(0)}$ to be our initial feasible solution to be
$ \xx^{(0)} = \min_{\AA\xx=\bb} \norm{\xx}^2_2.$
\end{definition} 
\begin{restatable}{lemma}{InitialPoint}
\label{lem:initialpoint}
For $\xx^{(0)}$ as defined in Definition \ref{start}, $\smallnorm{\xx^{(0)}}_p^p \leq m^{\nfrac{(p-2)}{2}} \opt$.
\end{restatable}
\noindent We are now ready to prove Theorem \ref{thm:IterativeRefinement}.
\begin{proof}
  Let $\widetilde{\Delta}$ denote the solution returned by the $\kappa$-approximate solver. We know that
  $\residual(\widetilde{\Delta}) \geq \frac{1}{k} \cdot
  \residual(\Delta^{\star})$. We have,
\[
  \residual(\widetilde{\Delta})  \geq \frac{1}{\kappa} \cdot \residual(\Delta^{\star}) \geq
  \frac{1}{\kappa} \residual(\xx - \xx^{\star})
  \\ \geq \frac{1}{\kappa}\left( \pnorm{\xx}^p - \pnorm{\xx^{\star}}^p \right)  = \frac{1}{\kappa} \left(\pnorm{\xx}^p - \opt \right).
\]
From Lemma \ref{lem:LambdaBound}, for
$\lambda = \left(\frac{p-1}{p4^p}\right)^{\frac{1}{\min\{1,p-1\}}} =
\Omega_p(1),$ we get,
\begin{equation*}
  \norm{\xx - \lambda \widetilde{\Delta}}_{p}^{p} \le \norm{\xx}_{p}^{p} -
  \lambda \residual(\widetilde{\Delta}) .
\end{equation*}
Combining the above two equations and subtracting $\opt$ from both sides gives us
\begin{align*}
  \pnorm{\xx- \lambda \widetilde{\Delta}}^p -  \opt 
  &\leq -\lambda \residual(\widetilde{\Delta})+\pnorm{\xx}^p - \opt\\
  & \leq - \frac{\lambda}{\kappa}\left(\pnorm{\xx}^p -\opt \right) +\left(\pnorm{\xx}^p -\opt \right)\\
  & \leq \left(1 - \frac{\lambda}{\kappa} \right)\left(\pnorm{\xx}^p
    -\opt \right).
\end{align*}
Using lemma \ref{lem:initialpoint} we get,
\[
\xx^{(t)} - \opt   \leq \left(1 - \frac{\lambda}{\kappa} \right)^t \left(\xx^{(0)} - \opt \right) \\ 
\leq \left(1 - \frac{\lambda}{\kappa} \right)^t \left(m^{\frac{p-2}{2}} - 1 \right) \opt. 
\]
\noindent Setting $t = O_p \left( \kappa \log\left(\frac{m}{\eps}\right)\right)$ gives us an $\eps$-approximate solution.
\end{proof}
\noindent This concludes the discussion on the analysis of Algorithm
\ref{alg:MetaAlgo}. In the following sections we move on to analyzing how to solve the
residual problem approximately.

\section{Solving the Residual Problem}
\label{sec:ResidualAlgo}

In this section, we give an algorithm that solves the
  residual problem to a constant approximation.  Combined with the
  iterative refinement scheme from
  Theorem~\ref{thm:IterativeRefinement}, we obtain the following
  result.
\begin{restatable}{theorem}{MainResult}
\label{thm:MainResult}
For $p \geq 2$, we can find an $\eps$-approximate solution to \eqref{eq:primal} in time 
\[\Otil_p\left((m+n)^{\omega + \frac{p-2}{3p-2}} \log^{2}
    \nfrac{1}{\eps} \right).\]
Here $\omega$ is the matrix multiplication constant.
\end{restatable}

\noindent Recall that the residual problem 
\[
\max_{\AA \Delta = 0} \gg^{\top}\Delta - \frac{p-1}{p2^p} \gammap \left(\tt,\Delta\right),
\]
has a linear term followed by the $\gammap$ function. Instead of
directly optimizing this function, we guess an approximate value of
the linear term, and for each such guess, we minimize the $\gammap$
function under this additional constraint. We can scale the problem so
that the optimum is at most $1.$ Finally, we can perturb $\tt$ so that
each $\tt_i$ lies in a polynomially bounded range without adding
significant error. Our final problem looks as follows,
\begin{align}
\label{eq:ScaledProblem}
\begin{aligned}
\min_{\Delta} &\quad  \gammap\left(\tt,\Delta\right)\\
& \AA \Delta = 0,\\
& \gg^{\top}\Delta = c,
\end{aligned}
\end{align}
 with $m^{-1/p} \leq \tt_i \leq 1, \forall i$.\\

\noindent  To sumarise, {\sc $\kappa$-Approx} (Algorithm
$\ref{alg:MainAlgo}$) formalizes this process and shows that we only
need to solve a logarithmic number of instances of the above program,
\eqref{eq:ScaledProblem} and solving each to a $\kappa$-approximation
gives a $\Omega_p\left(\kappa^{1/(\min\{2,p\}-1)}\right)$-approximate
solution to \eqref{eq:residual}. {\sc Gamma-Solver} (Algorithm
\ref{alg:FasterOracleAlgorithm}) solves problem \eqref{eq:ScaledProblem} to an $O_p(1)$ approximation. Therefore, using {\sc Gamma-Solver} as a subroutine for {\sc $\kappa$-Approx} we get an $O_p(1)$ approximate solution to \eqref{eq:residual}. Section \ref{sec:Scaling} gives an analysis for {\sc $\kappa$-Approx}. In Section \ref{sec:Oracle}, we give an oracle that is used in {\sc Gamma-Solver}. We give an analysis of {\sc Gamma-Solver} in Section \ref{sec:MWUAlgo}. Finally in Section \ref{sec:ProofMainTheorem}, we give a proof for Theorem \ref{thm:MainResult}.

\subsection{Equivalent Problems} \label{sec:Scaling}
\noindent In this section we prove the following theorem.
\begin{restatable}{theorem}{Algo2}
\label{thm:Algo2}
Procedure {\sc $\kappa$-Approx} (Algorithm \ref{alg:MainAlgo}) returns an $\Omega_p\left(\kappa^{\nfrac{1}{(\min\{2,p\}-1)}}\right)$-approximate solution to the residual problem given by \eqref{eq:residual}, by solving $O_p\left( \log \left(\frac{m}{\eps}\right)\right)$ instances of program \eqref{eq:ScaledProblem} to a $\kappa$-approximation.
\end{restatable}
\noindent The following lemmas will lead to the proof of the above theorem. The first lemma gives an upper and lower bound on the objective of \eqref{eq:residual}.
\begin{restatable}{lemma}{BoundResidual}
\label{lem:BoundResidual}
Let $p \in (1,\infty)$ and assume that our current solution $\xx$ is not an $\eps$-approximate solution. Let $\lambda$ be such that $\lambda^{\min\{1,p-1\}} = \frac{p-1}{p4^p}$. For some \[i \in \left[\log \left(\frac{\eps\smallnorm{\xx^{(0)}}_p^p}{m^{\nfrac{|p-2|}{2}}}\right), \log \left(\frac{\smallnorm{\xx^{(0)}}_p^p}{\lambda} \right)\right],\] $\residual(\Delta^{\star}) \in [2^{i-1},2^i)$ where $\Delta^{\star}$ is the solution of \eqref{eq:residual}.
\end{restatable}

\noindent We defer the proof to Appendix \ref{sec:ProofsSec5}. Lemma \ref{lem:BoundResidual} suggests that we can divide the range of the objective of our residual problem, $\residual$ into a logarithmic
number of bins and solve a {\it decision problem} that asks if the optimum belongs to the bin. The lemma guarantees that at least one of the decision problems will be feasible. The
following lemma defines the required {\it decision problems} and shows that solving these to a constant approximation is sufficient to get a constant approximate solution to \eqref{eq:residual}.
\begin{restatable}{lemma}{StepProblemSearch}
\label{lem:stepProblemSearch}
Let $p \in (1,\infty)$. Suppose $\residual(\Delta^{\star}) \in [2^{i-1},2^i)$ for some $i$ where $\Delta^{\star}$ is the solution of \eqref{eq:residual}. The following program is feasible:
\begin{align}
    \label{eq:BinarySearchProblems}
    \begin{aligned}
      \gamma_{p} \left( \tt , \Delta \right) &
      \le  \frac{p}{p-1}2^{i+p}, \\
      \gg^T \Delta & = 2^{i-1}, \\
      \AA \Delta & = 0.
    \end{aligned}
  \end{align}
 If $\Delta(i)$ is a $\beta$-approximate solution to 
  program \eqref{eq:BinarySearchProblems} for this choice of $i,$ then, we can pick $\mu \le 1$ such that the vector $\mu \Delta(i)$ is
  an $\Omega_{p}\left(\beta^{\frac{1}{\min\{p,2\} - 1}} \right)$-approximate solution to~\eqref{eq:residual}.
\end{restatable}
\noindent The proof can be found in Appendix \ref{sec:ProofsSec5}. We now scale down the objective of \eqref{eq:BinarySearchProblems} so that it is at most $1$. The next lemma shows what the scaled down problem looks like and how an approximate solution to the scaled down problem gives an approximate solution to \eqref{eq:BinarySearchProblems}. Again the proof of the lemma can be found in Appendix \ref{sec:ProofsSec5}.
\begin{restatable}{lemma}{Dual}
\label{lem:Dual}
Let $p \in (1, \infty)$. Let $i$ be such that \eqref{eq:BinarySearchProblems} is feasible. Let 
\[\textstyle
\hat{\tt}_j = \begin{cases}
m^{-\nfrac{1}{p}}  &\left(\frac{p-1}{p}\right)^{\nfrac{1}{p}}2^{\nfrac{-i}{p}-1}\tt_j \leq m^{\nfrac{-1}{p}},\\
1  &\left(\frac{p-1}{p}\right)^{1/p} 2^{-i/p-1}\tt_j \geq 1,\\
\left(\frac{p-1}{p}\right)^{\nfrac{1}{p}} 2^{\nfrac{-i}{p}-1}\tt_j  &\text{otherwise}.
\end{cases}
\] Note that $m^{-1/p} \leq \hat{\tt}_j \leq 1$. Then program \eqref{eq:ScaledProblem} with $\tt = \hat{\tt}$, and 
\[\cc =\left(\frac{2}{p}\right)^{1/2}\left(\frac{p-1}{p}\right)^{1/p} 2^{i \left(1-\frac{1}{p}\right) - 2},\] has $\opt \leq 1$.  
Let $\Delta^{\star}$ be a $\kappa$-approximate solution to \eqref{eq:ScaledProblem}. Then,
$\Delta = \left(\frac{p}{2}\right)^{1/2} \left(\frac{p}{p-1}\right)^{1/p}2^{1+ \nfrac{i}{p}} \Delta^{\star}$
is a $\Omega_p (\kappa)$- approximate solution to \eqref{eq:BinarySearchProblems}.
\end{restatable}

\begin{algorithm}
\caption{Approximate Solver}
\label{alg:MainAlgo}
 \begin{algorithmic}[1]
 \Procedure{$\kappa$-Approx}{$\AA, \bb, \xx,\lambda,\xx^{(0)}$}
\State $\AAhat \leftarrow \begin{bmatrix}
    \AA^{\top},
    \gg
  \end{bmatrix}^{\top}$ \\
 \For{$i \in \left[\log \left(\frac{\eps\smallnorm{\xx^{(0)}}_p^p}{m^{\nfrac{(p-2)}{2}}}\right), \log \left(\frac{\smallnorm{\xx^{(0)}}_p^p}{\lambda} \right)\right]$}
 \State $c(i) \leftarrow \left(\frac{2}{p}\right)^{1/2}\left(\frac{p-1}{p}\right)^{1/p} 2^{i \left(1-\frac{1}{p}\right) - 2}$
 \State For every $e$,

 \Statex{\qquad \qquad $\tt_e \leftarrow \max\{m^{-\frac{1}{p}},(\frac{p-1}{p})^{1/p}2^{-i/p-1}|\xx^{(t)}_e|\}$}
 \State For every $e$, $\tt_e \leftarrow \min\{1,\tt_e\}$
 \State $\cc \leftarrow  \begin{bmatrix}
    {\boldsymbol{{0}}}^{\top}, 
    c(i)
  \end{bmatrix}^{\top}$
 \State{$\Delta^{\star} \leftarrow$ {\sc Gamma-Solver}$(\AAhat,\cc,\tt )$}\quad \Comment{$\kappa$-approximation to \eqref{eq:ScaledProblem}}
 \State $\Delta^{(i)} \leftarrow \left(\frac{p}{2}\right)^{1/2} \left(\frac{p}{p-1}\right)^{1/p}2^{1+ \nfrac{i}{p}} \Delta^{\star}$
 \State $\beta \leftarrow \left( \frac{p}{2}\right)^{p/2} \frac{p2^{p+1}}{p-1} 2^i \kappa$
 \State $\mu^{(i)} \leftarrow \begin{cases}
\left(\frac{1}{2\beta p}\right)^{1/(p-1)} & \text{   if $p \leq 2$}\\
\frac{1}{4\beta} & \text{otherwise.}
\end{cases}$
\State $\Delta^{(i)} \leftarrow \mu^{(i)} \Delta^{(i)}$
 \EndFor
\Return $\lambda \cdot  \argmin_{\Delta^{(i)}} \smallnorm{\xx^{(t)} - \lambda \Delta^{(i)}}_p^p$
\EndProcedure 
 \end{algorithmic}
\end{algorithm}

\noindent We now prove Theorem \ref{thm:Algo2}.
\begin{proof}
Lemma \ref{lem:BoundResidual} suggests that there exists an index
\[
j \in \left[\log \left(\frac{\eps\smallnorm{\xx^{(0)}}_p^p}{m^{\nfrac{(p-2)}{2}}}\right), \log \left(\frac{\smallnorm{\xx^{(0)}}_p^p}{\lambda} \right)\right],
\]
 such that $OPT = \residual(\Delta^{\star}) \in [2^{j-1},2^j)$. Lemma \ref{lem:stepProblemSearch} implies that \eqref{eq:BinarySearchProblems} is feasible for index $j$. Suppose $\Delta^{(j)}$ is a $\kappa$-approximate solution to the scaled down problem \eqref{eq:ScaledProblem} for index $j$. Lemma \ref{lem:Dual} implies that $\tilde{\Delta}^{(j)} = \left(\frac{p}{2}\right)^{1/2} \left(\frac{p}{p-1}\right)^{1/p}2^{1+ \nfrac{i}{p}} \Delta^{(j)}$ is an $\Omega_p (\kappa)$ approximate solution to \eqref{eq:BinarySearchProblems} for index $j$. Lemma \ref{lem:stepProblemSearch} now implies that $\tilde{\Delta}^{(j)} =\mu \Delta^{(j)}$ is a $\Omega_{p}\left(\kappa^{\frac{1}{\min\{p,2\} - 1}} \right)$-approximation to our residual problem \eqref{eq:residual}. Now, the algorithm solves the scaled down problem for every $i$ and returns the $\tilde{\Delta}^{(i)} $ that when added to our current solution gives the minimum $\ell_p$-norm. It either chooses $\tilde{\Delta}^{(j)} $ or some other solution $\tilde{\Delta}^{(i)} $. In case it returns $\tilde{\Delta}^{(i)} $,
\begin{align*}
\|\xx\|^p_p - \|\xx -\lambda \tilde{\Delta}^{(i)} \|_p^p & \geq \|\xx\|^p_p - \|\xx -\lambda \tilde{\Delta}^{(j)} \|_p^p \\
& \geq \lambda \cdot \residual(\tilde{\Delta}^{(j)}), \text{Lemma \ref{lem:LambdaBound}}, \\
& \geq \lambda \cdot \Omega_{p}\left(\kappa^{- \frac{1}{\min\{p,2\} - 1}} \right) \residual(\Delta^{\star})\\
& = \Omega_{p}\left(\kappa^{- \frac{1}{\min\{p,2\} - 1}} \right) \residual(\Delta^{\star}).
\end{align*}
From Lemma \ref{lem:LambdaBound} we know,
\[
\residual(\lambda \tilde{\Delta}^{(i)}) \geq \|\xx\|^p_p - \|\xx -\lambda \tilde{\Delta}^{(i)} \|_p^p. 
\]
We thus have $\residual(\lambda \tilde{\Delta}^{(i)} ) \geq \Omega_p\left(\kappa^{-\frac{1}{\min\{p,2\} - 1}} \right)\opt$, implying $\lambda \cdot \tilde{\Delta}^{(i)}$ is also a $\Omega_p\left(\kappa^{\frac{1}{\min\{p,2\} - 1}} \right)$ approximate solution as required.
\end{proof}

\noindent  It remains to solve problems of the form \eqref{eq:ScaledProblem} up to a $\kappa$-approximation. Recall that these problems look like,
\begin{align*}
\min_{\Delta} &\quad  \gammap\left(\tt,\Delta\right)\\
& \AAhat \Delta = \dd,
\end{align*}
and satisfy %
$OPT \leq 1$, and
  $m^{-1/p} \leq \tt_j \leq 1, \forall j$.

\subsection{Oracle}
\label{sec:Oracle}
Our approach follows the format of the approximate max-flow algorithm
by Christiano et al.~\cite{ChristianoKMST10}.
Specifically, we use a variant of multiplicative weights update to converge
to a solution with small $\gamma_{p}(t, \Delta)$.
The multiplicative weights update scheme repeatedly updates a set of weights
$\ww$ using partial, local solutions computed based on these weights.
The Christiano et al. algorithm can be viewed as picking these weights from
the gradients of the soft-max function on flows.
We will adapt this routine by showing that $\ww$'s chosen from the gradient
of $\gamma_{p}(t, \Delta)$ also suffices for approximately minimizing the
problem stated in~\ref{eq:ScaledProblem}.

The subroutine that this algorithm passes the $\ww$ onto is commonly
referred to as an oracle.
An oracle needs to compute a solution with both small dot-product against
$\ww$, and small \emph{width}, which is defined as the maximum value of
an entry.
In such an oracle, the dot product condition is the hard constraint, in that
the final approximation factor of the solution is directly related to the
value of these dot products.
The width, on the other hand, only affects the overall iteration count/
running time, and can even be manipulated/improved algorithmically.
Therefore we first need to define and show a good upper bound on the objective
of the optimization problem solved within the oracle.

Formally, our oracle subroutine Algorithm \ref{alg:oracle} takes as
input some affine constraints and vector of weights $\ww$.
It first computes a vector of non-negative weights $\rr$,
and then returns a minimizer to the following optimization problem
\begin{align}
\label{eq:oracleprog}
  \Delta = \argmin_{\Delta \in \mathbb{R}^m} \quad
  & \sum_e
  \rr_e \Delta_e^2\\
\nonumber
  \text{s.t. } \quad & \AAhat \Delta = \dd .
\end{align}
Appendix \ref{sec:L2Solver} contains an algorithm that solves such problems efficiently.

\begin{algorithm}
\caption{Oracle}
\label{alg:oracle}
 \begin{algorithmic}[1]
 \Procedure{\textsc{Oracle}}{$\AAhat, \dd, \ww, \tt$}
\State $\rr_e \leftarrow \left( m^{1/p}\tt_e\right)^{p-2} + \ww_e^{p-2}$
\label{algline:resistance}

\State Compute the $\Delta$ using resistances $\rr_e$
that satisfies solves the following optimization problem
\begin{align*}
  \Delta = \argmin_{\Delta'} \quad
  & \sum_e
  \rr_e {\Delta^{'}}^2_e\\
  \text{s.t. } \quad & \AAhat \Delta' = \dd 
\end{align*}
 \State\Return $ \Delta$
 \EndProcedure 
 \end{algorithmic}
\end{algorithm}

Let us now look at some properties of the solution returned by the oracle. Note that the objective of our problem \eqref{eq:ScaledProblem} is at most $1$. This implies that we have $\Delta^{}$ such that
\begin{tight_itemize}
\item $\sum_{e} (\Delta^*_e)^2 \tt_{e}^{p-2} \leq 1$,
\item $\sum_{e} \abs{\Delta^*_e}^p \leq 1$, or $\norm{\Delta^*}_p \leq 1$.
\end{tight_itemize}

 We next look at some relations on the weights and resistances. The following lemma is a simple application of H\"{o}lder's
inequality. Its proof is given in Appendix \ref{sec:ProofsSec5}.
\begin{restatable}{lemma}{BoundOpt}
\label{lem:BoundOpt}
 Let $p\geq 2$. For any set of weights $\ww$ on the edges,
$  \sum_e \ww_e^{p-2} (\Delta^*_e)^2 \leq \pnorm{\ww}^{p-2}.$
\end{restatable}

\begin{lemma}
\label{lem:Oracle}
Let $p\geq 2$. For any $\ww$, let $\Delta$ be the electrical flow computed with
respect to resistances
\[\rr_e \defeq \paren{m^{\nfrac{1}{p}} \tt_e}^{p-2} + \ww_e^{p-2}, \]
and demand vector $\dd.$

Then the following hold,
\begin{tight_enumerate}
\item \label{item:oracle:l2}
    $\sum_e
  \Delta_e^2 \leq \sum_e \rr_e
  \Delta_e^2 \leq m^{\frac{p-2}{p}} + \pnorm{\ww}^{p-2},$
\item $\textstyle \sum_e \abs{\Delta_e} \abs{\gamma' (m^{\nfrac{1}{p}}\tt_e, \ww_e)}  \leq  p  \paren{\sum_e \gammap(m^{\nfrac{1}{p}}\tt_e, \ww_e) }^{\frac{p-1}{p}}+ pm^{\frac{p-2}{2p}} \paren{\sum_{e} \gammap(m^{\nfrac{1}{p}}\tt_e, \ww_e)}^{\frac{1}{2}}.$
  \end{tight_enumerate}
\end{lemma}
\begin{proof}
Since $\Delta$ is the electrical flow, 
\begin{align*}
\sum_e \rr_e \Delta_e^2 \leq \sum_e \rr_e (\Delta^*_e)^2.
\end{align*}
We have,
\begin{align*}
  \sum_e \rr_e \Delta_e^2 \leq \sum_e \rr_e (\Delta^*_e)^2
  & = \sum_e (m^{\nfrac{1}{p}} \tt_e )^{p-2}(\Delta^*_e)^2 + \sum_e \ww_e^{p-2}
    (\Delta^*_e)^2\\
  &\leq m^{\frac{p-2}{p}} + \norm{\ww}^{p-2}_p , \text{follows from lemma \ref{lem:BoundOpt} and the fact $\sum_{e} (\Delta^*_e)^2 \tt_{e}^{p-2} \leq 1$}.
\end{align*}
Finally, using
$\rr_e \ge \left( m^{\nfrac{1}{p}} \tt_e \right)^{p-2} \ge 1,$ we have
$\sum_e \Delta_e^{2} \le \sum_e \rr_e \Delta_e^{2},$ completing part
\ref{item:oracle:l2}.

Now we know that,
\[
  \abs{\gamma'(m^{\nfrac{1}{p}}\tt_e,\ww_e)} =
  \begin{cases}
      p(m^{\nfrac{1}{p}}\tt_e)^{p-2}\ww_e & \text{if } |\ww_e| \leq m^{\nfrac{1}{p}}\tt_e,\\
      p |\ww_e|^{p-2} \ww_e & \text{otherwise }.
          \end{cases}
\]
Using Cauchy Schwarz's inequality,
\begin{align*}
  \left(\sum_e \abs{\Delta_e}
    \abs{\gamma'(m^{\nfrac{1}{p}}\tt_e,\ww_e)}\right)^2
  = &
  \left(\sum_e  p \abs{\Delta_e}
    \abs{\max(m^{1/p}\tt_e, \ww_e)}^{p-2} \ww_e^2\right)^2\\
    \leq &
  p^{2} \left(\sum_{e} \max(m^{1/p}\tt_e, \ww_e)^{p-2} \ww_e^2\right) \left(\sum_e \max(m^{1/p}\tt_e, \ww_e)^{p-2}\Delta_e^2\right) \\
  \leq & p^2 \gammap(m^{\nfrac{1}{p}}\tt,\ww) \sum_e \max(m^{\nfrac{1}{p}}\tt_e, \abs{\ww_e})^{p-2} \Delta_e^2
\end{align*}
Combining the two cases we have,
\begin{align*}
  \sum_e \abs{\Delta_e} \abs{\gamma'(m^{\nfrac{1}{p}}\tt_e,\ww_e)}
  & \leq  p\sqrt{\gammap(m^{\nfrac{1}{p}}\tt,\ww) \sum_e \max(m^{\nfrac{1}{p}}\tt_e, \abs{\ww_e})^{p-2} \Delta_e^2}\\ 
  & \leq  p\sqrt{\gammap(m^{\nfrac{1}{p}}\tt,\ww) \sum_e \rr_e \Delta_e^2}\\
  & \leq  p\sqrt{\gammap(m^{\nfrac{1}{p}}\tt,\ww)}  \sqrt{ m^{\frac{p-2}{p}} + \pnorm{\ww}^{p-2}}\\
  & \leq  p  m^{\frac{p-2}{2p}} \sqrt{\gammap(m^{\nfrac{1}{p}}\tt,\ww)}   + p \sqrt{\gammap(m^{\nfrac{1}{p}}\tt,\ww)}   \pnorm{\ww}^{\frac{p-2}{2}}\\
  & =  pm^{\frac{p-2}{2p}}\gammap(m^{\nfrac{1}{p}}\tt,\ww)^{\frac{1}{2}} + p \gammap(m^{\nfrac{1}{p}}\tt,\ww)^{\frac{p-1}{p}},
\end{align*}
where the last line uses $\pnorm{x}^{p} \le \gammap(m^{\nfrac{1}{p}} \tt, \ww)$
for any $\tt.$
\end{proof}

\subsection{The Algorithm}
\label{sec:MWUAlgo}

\label{subsec:WidthReduction}

Next, we integrate this oracle into the overall algorithm that
repeatedly adjusts the weights.
As with the use of electrical flow oracles for approximate
max-flow~\cite{ChristianoKMST10}, the convergence of such a scheme
depends on the maximum values in the $\Delta$ returned by the oracle.
However, because the overall objective is now a $p$-norm, the exact
term of importance is actually the $p$-norm of $\Delta$.
Up to this discrepancy, we follow the algorithmic template 
from~\cite{ChristianoKMST10} by making an update when $\norm{\Delta}_{p}^{p}$
is small, and make progress via another potential function otherwise.

In the cases where we do not take the step due to entries with large
values, we show significant increases in an additional potential function,
namely the objective of the quadratic minimization problem inside the
oracle (Algorithm~\ref{alg:oracle}).
However, the less graduate update schemes related to $p$-norms makes it
no longer sufficient to update only the weight corresponding to the
entry with maximum value.
Furthermore, there may be entries with large values, whose corresponding
resistances are too large for us to afford increasing.
We address this by a scheme where we update an entry only if its value
is larger than some threshold $\rho$, and that its resistance is at most
another threshold $\beta$.
Specifically, we show that for an appropriate choice of $\beta$ and $\rho$,
such updates both do not change the primary potential function
(related to $\gamma_{p}(t, \xx)$) by too much (in Lemma~\ref{lem:ReduceWidthGammaPotential}),
and increases the secondary potential function (the objective of the
quadratic minimization problem) significantly whenever $\norm{\Delta}_{p}^{p}$
is large (in Lemma~\ref{lem:ReduceWidthElectricalPotential}).
Pseudocode of this scheme is in Algorithm~\ref{alg:FasterOracleAlgorithm}.

\begin{algorithm}
\caption{Algorithm for the Scaled down Problem}
\label{alg:FasterOracleAlgorithm}
 \begin{algorithmic}[1]
 \Procedure{\FasterGammaApprox}{$\AA', \cc, \tt$}
 \State $\ww^{(0,0)}_e \leftarrow 0$
 \State $\xx \leftarrow 0$
 \State $\rho \leftarrow \tilde{\Theta}_p\left(m^{\frac{(p^2-4p+2)}{p(3p-2)}}\right)$\Comment{width parameter}%
\label{algline:defrho}
 \State $\beta \leftarrow \tilde{\Theta}_p \left(m^{\frac{p-2}{3p-2}}\right)$\Comment{resistance threshold}%
\State $\alpha \leftarrow \tilde{\Theta}_p\left(m^{-\frac{p^2-5p+2}{p(3p-2)}}\right)$\Comment{step size}%
\State $\tau \leftarrow \tilde{\Theta}_p\left(m^{\frac{(p-1)(p-2)}{(3p-2)}}\right)$\Comment{$\ell_p$ energy threshold}%
\label{algline:defalpha}
\State $T \leftarrow \alpha^{-1} m^{1/p} = \tilde{\Theta}_p \left(m^{\frac{p-2}{3p-2}}\right)$%
\State{$i \leftarrow 0, k \leftarrow 0$}
\While{$i < T$} 
\State $\Delta = \textsc{Oracle}(\AA',\cc,\ww^{(i,k)},\tt)$
\label{algline:InvokeOracle}
\If{$\norm{\Delta}_{p}^p \leq \tau$} \hfill \Comment{flow step}
\label{algline:CheckWidth}
\State $\ww^{(i+1,k)} \leftarrow \ww^{(i,k)} + \alpha \abs{\Delta}$
\label{algline:LowWidth}
\State $\xx \leftarrow \xx +  \alpha \Delta$
\State $ i \leftarrow i+1$ 
\Else \hfill \Comment{width reduction step}
\State{For all edges $e$ with $|\Delta_e| \geq \rho$ and $\rr_e \leq \beta$\label{lin:WidthReduceEdge}}
\Statex{\qquad \qquad \qquad $\ww_e^{(i,k+1)} \leftarrow 4^{\frac{1}{p-2}} \max (  m^{1/p}  \tt_e,  \ww_e^{(i,k)} )$} %
\State{$ k \leftarrow k+1$ \label{lin:widthReductionStepIncr}}
\EndIf
\EndWhile
 \State\Return $m^{- \frac{1}{p}} \xx$
 \EndProcedure 
 \end{algorithmic}
\end{algorithm}

\begin{theorem}\label{lem:FasterAlgo}
Let $p\geq 2$.  Given a matrix $\AAhat$ and vectors $\xx$ and $\tt$ such that $\forall e, m^{-1/p} \leq \tt_e \leq 1 $, Algorithm \ref{alg:FasterOracleAlgorithm} uses $O_p\left(m^{\frac{p-2}{(3p-2)}}\left(\log \left(\frac{m \norm{\dd}^2_2}{\|\AAhat\|^2}\right)\right)^{\frac{p}{3p-2}}\right)$ calls to the oracle and returns a vector $\xx$ such that 
$\AAhat \xx = \dd,$ and $\gammap(\tt,\xx) = O_p(1)$.
\end{theorem}

\subsection*{Analysis of Potentials.}
We define the following potential function for the analysis of our algorithm.

\begin{definition}
Let $\Phi$ be the potential function defined as
\[
  \Phi\left(\ww^{\left( i \right)} \right)
  \defeq \gammap\left( m^{\nfrac{1}{p}}
    \tt,
      \ww^{\left( i \right)}\right).
\]
\end{definition}
\noindent Initially, since we start with $\ww^{(0)} = 0$, we have $\Phi(\ww^{(0)}) = 0.$ Observe that in the algorithm, 
we update the potentials in both the flow step and the width reduction step whereas we update the solution only in the flow step. It is easy to see that we always have $\ww^{(i,k)} \ge \abs{\xx^{(i,k)}}.$ 

We next bound the potential. In addition, we track the energy  of the electrical flow in the
network with resistances $\rr.$  Let $\energy{\rr}$ denote the minimum
of routing $\dd$ with resistances $\rr$:
\begin{align}
\energy{\rr}
\defeq
\min_{\Delta: \AA' \Delta = \dd}
\sum_{e} \rr_e \Delta_e^2.
\end{align}
Note that this energy is equal to the energy calculated using the $\Delta$ obtained in the solution of \eqref{eq:oracleprog}.

\noindent \paragraph{Notation. }
We overload notation for $\energy{i,k}$ to denote
$\energy{\rr^{(i,k)}}.$

\noindent Our proof of Theorem~\ref{lem:FasterAlgo} will be based two main parts:
\begin{enumerate}
\item Provided the total number of width reduction steps, $K$,
  is not too big, then $\Phi(T,K)$ is small.
This in turn upper bounds cost of the approximate solution
$m^{-1/p} \xx$.
\item Showing that $K$ cannot be too big, because each width reduction
  step cause large growth in $\energy{\cdot}$, while we can bound the total
  growth in $\energy{\cdot}$ by relating it to $\Phi(\cdot)$.
\end{enumerate}
We start by observing that when we when increase the weight $\ww_e$ of an edge
during a width reduction step, this has the effect of at least
doubling the resistance $\rr_e$.
\noindent Recall,
\[\rr_e^{(i,k)} \defeq \paren{m^{\nfrac{1}{p}} \tt_e}^{p-2} +
  \left(\ww_e^{(i,k)}\right)^{p-2}.
\]
Now,
\begin{align}
\label{eq:resistanceDoubles}
  \frac{\rr_{e}^{(i,k+1)}}{\rr_{e}^{(i,k)}}
  & = \frac{ \paren{m^{\nfrac{1}{p}} \tt_e}^{p-2} + \left(\ww_e^{(i,k+1)}\right)^{p-2}}
    { \paren{m^{\nfrac{1}{p}} \tt_e}^{p-2} + \left(\ww_e^{(i,k)}\right)^{p-2}}  = \frac{ \paren{m^{\nfrac{1}{p}} \tt_e}^{p-2} +   4 \max\left\{
    m^{\nfrac{1}{p}} \tt_e, \ww_e^{(i,k)} \right\}^{p-2}}
    { \paren{m^{\nfrac{1}{p}} \tt_e}^{p-2} + \left(\ww_e^{(i,k)}\right)^{p-2}}   \ge 2.
\end{align}
Meanwhile, the resistance does not grow by a factor larger than 4:
\begin{align}
\label{eq:resistanceGrowthUpperBound}
  \frac{\rr_{e}^{(i,k+1)}}{\rr_{e}^{(i,k)}}
  & = \frac{ \paren{m^{\nfrac{1}{p}} \tt_e}^{p-2} +   4 \max\left\{
    m^{\nfrac{1}{p}} \tt_e, \ww_e^{(i,k)} \right\}^{p-2}}
    { \paren{m^{\nfrac{1}{p}} \tt_e}^{p-2} + \left(\ww_e^{(i,k)}\right)^{p-2}}   \leq 4.
\end{align}
We next show through the following lemma that the $\Phi$ potential does not increase
too rapidly. The proof is through induction and can be found in Appendix \ref{sec:ControllingGammaPotential} .
\begin{restatable}{lemma}{ReduceWidthGammaPotential}
  \label{lem:ReduceWidthGammaPotential}
  After $i$ \emph{flow} steps, and $k$ width-reduction steps,
  provided
  \begin{enumerate}
  \item 
\label{enu:pPowerStep}
$\alpha^p \tau \leq \alpha m^{\frac{p-1}{p}}$,
(controls $\Phi$ growth in flow-steps)
  \item 
\label{enu:widthStepCondition}
$k \leq \rho^2 m^{2/p} \beta^{-\frac{2}{p-2}}$ ,
(acceptable number of width-reduction steps)
  \end{enumerate}
  the potential $\Phi$ is bounded as follows:
  \begin{align*}
    \Phi(i,k) \leq
    \left(
    {p^{2}2^{p}\alpha i} + m^{\nfrac{1}{p}} \right)^{p}
    \exp{\left(O_p(1) \frac{k}{\rho^2 m^{2/p} \beta^{-\frac{2}{p-2}}}
    \right)}.
  \end{align*}
\end{restatable}
\noindent We next wish to prove that in each width-reduction step, the electrical
energy $\energy{\cdot}$ goes up significantly. For this, we will
use the following Lemma which is proven in Appendix \ref{sec:ckmstResIncreaseProof}. It generalizes Lemma 2.6 of~\cite{ChristianoKMST10} to arbitrary weighted $\ell_2$ regression problems, and directly measures
the change in terms of the electrical energy of the entries modified.
\begin{restatable}{lemma}{ckmstResIncrease}
  \label{lem:ckmst:res-increase}
Assuming the program~\eqref{eq:oracleprog} is feasible,
let $\Delta$ be an be a solution to the optimization
problem~\eqref{eq:oracleprog} with weights $\rr$.
Suppose we increase the resistance on each entry to get $\rr'$
Then,
\[
\energy{\rr'}
\ge
\exp \left( \frac{\sum_{e}
\min \left\{1, \frac{\rr'_e - \rr_e}{\rr_e}\right\}
\rr_e \Delta_e^2}{2\energy{\rr}} \right)
\energy{\rr}.
\]
\end{restatable}
This statement also implies the form of the lemma that concerns
increasing the resistances on a set of entries uniformly~\cite[Lemma
2.6]{ChristianoKMST10}.

The next lemma gives a lower bound on the energy in iteration $0$, i.e., when we start, and an upper bound on the energy at each step.

\begin{lemma}
  \label{lem:ElectricalPotentialStartFinishBounds}
  Initially, we have,
 \[
\energy{\rr^{\left(0,0\right)}}
  \ge \frac{\norm{\dd}_{2}^2}{\norm{\AA}^2},
\]
where $\norm{\AA}$ is the operator norm, or
maximum singular value of $\AA$. Let us call this ratio $L$.
 Moreover, at any step $i,$ we have,
\[
  \energy{\rr^{(i,k)}} 
  \le m^{\frac{p-2}{p}} + \Phi(i,k)^{\frac{p-2}{p}}.
\]  
\end{lemma}
\begin{proof}
For the lower bound in the initial state,
recall that we scale the problem such that $\opt=1,$
and $\tt_{e} \ge m^{\nfrac{1}{p}}.$
Initially we have, $\rr^{(0,0)}_{e} = (m^{\nfrac{1}{p}} \tt_{e})^{p-2} \ge 1.$
This means for any solution $\Delta$, we have
\[
\sum_{e} \rr^{\left( 0, 0 \right)} \Delta_e^2
\geq
\norm{\Delta}_2^2.
\]
On the other hand, because
\[
\norm{\AA \Delta}_{2}
\leq
\norm{\AA}_{2} \norm{\Delta}_2,
\]
we get
\[
\norm{\Delta}_2 \geq \frac{\norm{\dd}_2}{\norm{\AA}_2},
\]
upon which squaring gives the lower bound on $\Psi(\rr^{(0, 0)})$.

\noindent For the upper bound, Lemma~\ref{lem:Oracle} implies that
  \[
    \energy{\rr^{(i,k)}} \le m^{\frac{p-2}{p}} + \pnorm{\ww}^{p-2} \le
    m^{\frac{p-2}{p}} + \Phi(i,k)^{\frac{p-2}{p}}.
  \]
\end{proof}

The next Lemma says that the three assumptions (stated in the statement of the Lemma)
can be used to ensure that the potential $\energy{\cdot}$ grows quickly with each width
reduction step, and that flow steps do not cause the potential to shrink.

\begin{restatable}{lemma}{ReduceWidthElectricalPotential}
\label{lem:ReduceWidthElectricalPotential}
  Suppose at step $(i,k)$ 
  we have $\norm{\Delta}_{p} >\tau$
  so that we perform a width reduction step
  (line~\ref{lin:widthReductionStepIncr}).
  If
  \begin{enumerate}
  \item 
\label{eq:parametersEnsuringElectricalPotentialGrowth0}
$\Phi(i,k) \leq O_p(1) m$,
  \item 
\label{eq:parametersEnsuringElectricalPotentialGrowth1}
$\tau^{2/p} \geq 2\Omega_p(1)  \frac{m^{\frac{p-2}{p}}}{\beta}$, and
\item 
\label{eq:parametersEnsuringElectricalPotentialGrowth2}
$\frac{\tau}{10} \geq \rho^{p-2} m^{\frac{p-2}{p}}$.
\end{enumerate}
Then
\[
  \energy{i,k+1} \geq \energy{i,k} \left(1+\Omega_p(1) \frac{\tau^{2/p}}{m^{\frac{p-2}{p}}}\right).
\]
Furthermore, if at $(i,k)$ 
  we have $\norm{\Delta}_{p} \leq \tau$
  so that we perform a flow step,
then
\[
  \energy{i+1,k} \geq \energy{i,k}.
\]
\end{restatable}

\begin{proof}
It will be helpful for our analysis to split the index set into three
disjoint parts:
\begin{itemize}
\item $S =  \setof{e : \abs{\Delta_e} \leq \rho }$
\item $H = \setof{e : \abs{\Delta_e} > \rho \text{ and } \rr_e \leq \beta }$
\item $B = \setof{e : \abs{\Delta_e} > \rho \text{ and } \rr_e > \beta }$.
\end{itemize}
Firstly, we note 
\begin{align*}
\sum_{e \in S} \abs{\Delta_e}^{p}
 \leq
\rho^{p-2} \sum_{e \in S} \abs{\Delta_e}^{2} 
 \leq
\rho^{p-2}  \sum_{e \in S} \rr_e \abs{\Delta_e}^{p} 
 \leq
\rho^{p-2} O_p(1) m^{(p-2)/p}.
\end{align*}
hence, using Assumption~\ref{eq:parametersEnsuringElectricalPotentialGrowth2}
\begin{align*}
\sum_{e \in H \union B } \abs{\Delta_e}^p
 \geq \sum_{e } \Delta_e^p - \sum_{e \in S} \abs{\Delta_e}^{p}
 \geq \tau - \rho^{p-2} m^{\frac{p-2}{p}}
\geq 9 \tau.
\end{align*}
This means,
\[
\sum_{e \in H \union B} \Delta_e^2
\geq 
\left(\sum_{e \in H \union B} \Delta_e^p\right)^{p/2} \geq \Omega_p(1) \tau^{2/p}.
\]
Secondly we note that, using Assumption \eqref{eq:parametersEnsuringElectricalPotentialGrowth0} and 
Lemma  \ref{lem:Oracle}, we have
\[
\sum_{e \in B} \Delta_e^2
\leq \beta^{-1} \sum_{e \in B} \rr_e \Delta_e^2 
\leq \beta^{-1} O_p(1) m^{\frac{p-2}{p}}.
\] 
So then, using Assumption~\ref{eq:parametersEnsuringElectricalPotentialGrowth1},
\begin{align*}
\sum_{e \in H } \Delta_e^2
  =
\sum_{e \in H \union B } \Delta_e^2
-
\sum_{e \in B} \Delta_e^2  \geq 
\Omega_p(1) \tau^{2/p} -\beta^{-1} O_p(1) m^{\frac{p-2}{p}} \geq 
\Omega_p(1) \tau^{2/p}
.
\end{align*}
As $\rr_e \geq 1$, this implies $\sum_{e \in H } \rr_e \Delta_e^2\geq \Omega_p(1) \tau^{2/p}$ .

From Lemma~\ref{lem:ElectricalPotentialStartFinishBounds} and
Assumption~\ref{eq:parametersEnsuringElectricalPotentialGrowth0} 
we have
  \[
    \energy{i,k}
\leq
O_p(1) m^{(p-2)/p}.
  \]
So then, combining our last two observations, and applying
Lemma~\ref{lem:ckmst:res-increase}, we get
\[
  \energy{i,k+1} \geq \energy{i,k} \left(1+\Omega_p(1)
    \frac{\tau^{2/p}}{m^{\frac{p-2}{p}}}\right)
.
\]
Finally, for the ``flow step'' case, by applying
Lemma~\ref{lem:ckmst:res-increase} with $H$ as the whole set of indices,
$\delta = 1$ and $\gamma = 1$, we get that as the resistances only
increase,
\[
\energy{i+1,k} \geq \energy{i,k}. 
\]
\end{proof}
\noindent We are now ready to prove Theorem \ref{lem:FasterAlgo}.

\subsection*{Proof of Theorem \ref{lem:FasterAlgo}}
\begin{proof}
We first observe that our parameter choices in the
Algorithm~\ref{alg:FasterOracleAlgorithm} 
satisfy Assumption~\ref{enu:pPowerStep}
of
Lemma~\ref{lem:ReduceWidthGammaPotential}, 
namely, we can choose the parameters $\alpha$ and $\tau$ s.t. 
\begin{itemize}
\item 
  $\alpha \leftarrow
 \Theta_p\left(m^{-\frac{p^2-5p+2}{p(3p-2)}} \left(\log \left(\frac{m}{L}\right)\right)^{\frac{-p}{3p-2}}\right)$,
\item
  $\tau \leftarrow \Theta_p\left(m^{\frac{(p-1)(p-2)}{(3p-2)}}\left(\log \left(\frac{m}{L}\right)\right)^{\frac{p(p-1)}{3p-2}}\right)$,
\end{itemize}
while ensuring 
$\alpha^p \tau \leq \alpha m^{\frac{p-2}{p}}$.
This means by Lemma~\ref{lem:ReduceWidthGammaPotential},
that if the Algorithm completes after taking $T = \alpha^{-1} m^{1/p}$
flow steps and $K \leq \Omega_p(1) \rho^2 m^{2/p} \beta^{-\frac{2}{p-2}}$,
when it returns, we have 
\[\Phi(T,K) \le  m  \left(
    p^{2}2^{p} + 1\right)^{p}
    e^{1}
    \leq 
    O_p(1) m
   ,\]
This means that the algorithm returns $m^{- \frac{1}{p}} \xx$ with 
\begin{align*}
\gamma(\tt, m^{- \frac{1}{p}} \xx)
 =
\frac{1}{m} \gamma(m^{\frac{1}{p}} \tt, \xx) 
\leq
\frac{1}{m} \gamma(m^{\frac{1}{p}} \tt, \ww^{(T,K)})  =\frac{1}{m} \Phi(T,K) 
\leq
O_p(1).
\end{align*}
Note the only alternative is that the algorithm takes more than
$\Omega_p(1) \rho^2 m^{2/p} \beta^{-\frac{2}{p-2}}$ width
reduction steps (and possibly infinitely many such steps, hence never
terminating).

We will now show this cannot happen, by deriving a contradiction from
the assumption that the algorithm takes a width reduction step
starting from step $(i,k)$ where $i < T$ and  $k = \rho^2 m^{2/p}
\beta^{-\frac{2}{p-2}}$.

Since the conditions for Lemma~\ref{lem:ReduceWidthGammaPotential}
hold for all preceding steps, we must have 
$\Phi(i,k) \leq O_p(1) m$.

Additionally, we note that our parameter choice of 
$\beta = \Theta_p \left(m^{\frac{p-2}{3p-2}}\left(\log \left(\frac{m}{L}\right)\right)^{-\frac{2(p-1)}{3p-2}}\right)$
and
$\rho = \Theta_p\left(m^{\frac{(p^2-4p+2)}{p(3p-2)}}\left(\log \left(\frac{m}{L}\right)\right)^{\frac{p(p-1)}{(p-2)(3p-2)}}\right)$ 
along with our choice of $\tau$ (see above), ensures that
\[
  \tau^{2/p} \geq 2\Omega_p(1)  \frac{m^{\frac{p-2}{p}}}{\beta}
\text{ and }
\frac{\tau}{10} \geq \rho^{p-2} m^{\frac{p-2}{p}}
.
\]
This means that at every step $(j,l)$ preceding the current step,
the conditions of Lemma~\ref{lem:ReduceWidthElectricalPotential} are
satisfied, so we can prove by a simple induction that 
\begin{align*}
\energy{i,k} &\geq \energy{0,0} \left(1+\Omega_p(1)
  \frac{\tau^{2/p}}{m^{\frac{p-2}{p}}}\right)^{k}  >
\energy{0,0} \exp\left( \Omega_p(1)
  \frac{\tau^{2/p}}{m^{\frac{p-2}{p}}} k \right)
.
\end{align*}
Since our parameter choices ensure
$\Omega_p(1) \frac{\tau^{2/p}}{m^{\frac{p-2}{p}}} k > \Theta_p\left(\frac{m}{L}\right)$
this means
\[
\energy{i,k}
>
\energy{0,0} \cdot \Theta_p\left(\frac{m}{L}\right).
\]
But this contradicts
Lemma~\ref{lem:ElectricalPotentialStartFinishBounds}, since this
Lemma, combined with
$\Phi(i,k) \leq O_p(1) m$ gives
\[
\frac{\energy{i,k}}{\energy{0,0}}
\leq O_p\left(m^{\frac{p-2}{p}} \right).
\]
From this contradiction, we conclude that we never have more than
$K = \Omega_p(1) \rho^2 m^{2/p} \beta^{-\frac{2}{p-2}}$ width
reduction steps.

Now we observe that the total number of oracle calls in the algorithm
is bounded by
\[
T + K \leq  \Theta_p \left(m^{\frac{p-2}{3p-2}}\left(\log \left(\frac{m}{L}\right)\right)^{\frac{p}{3p-2}}\right).
\]
\end{proof}
\noindent This concludes the analysis of our algorithm.

\subsection{Proof of Theorem \ref{thm:MainResult}}
\label{sec:ProofMainTheorem}
\begin{proof}
Theorem \ref{lem:FasterAlgo} implies that we can solve Program
\eqref{eq:ScaledProblem} using Algorithm
\ref{alg:FasterOracleAlgorithm} to get an $O_p(1)$-approximate
solution in $\Otil_p \left(m^{\frac{p-2}{3p-2}}\right)$ calls to the
Oracle. Implementing the Oracle requires solving a linear system, and
hence can be implemented in in $O(m+n)^{\omega}$ time where $\omega$
is the matrix multiplication constant (see the Appendix for a proof). Thus, we can find an $O_p(1)$-approximate solution to \eqref{eq:ScaledProblem} in total time 
\[
\Otil_p \left((m+n)^{\omega + \frac{p-2}{3p-2}}\right).
\]
Now, Theorem \ref{thm:Algo2} implies that we can find an $O_p(1)$-approximate solution to the residual problem \eqref{eq:residual} in total time,
\[
\Otil_p \left((m+n)^{\omega + \frac{p-2}{3p-2}} \log \nfrac{1}{\eps}\right).
\]
Finally using Theorem \ref{thm:IterativeRefinement} we can conclude that we have an $\eps$-approximate solution to \eqref{eq:primal} in $\Otil_p\left(\log \frac{1}{\eps}\right)$ calls to a $O_p(1)$-approximate solver to the residual problem \eqref{eq:residual}. This gives us a total running time of,
\[
  \Otil_p \left((m+n)^{\omega + \frac{p-2}{3p-2}} \log^{2}
    \nfrac{1}{\eps} \right).
\]
\end{proof}
We now have a complete algorithm for the $p$-norm regression problem that gives an $\eps$-approximate solution.

\section{Speedups for General Matrices via. Inverse Maintenance}
\label{sec:InvMaintain}

If $\AA$ is an explicitly given, $m \times n$, matrix, we need to solve
the quadratic minimization problem at each step. This can be solved
 via a linear systems solve in the matrix
\[
\AA^{\top} \diag{\rr}^{-1} \AA.
\]
which takes $O((m + n)^{\omega})$, where $\omega$ is the
matrix multiplication constant.
This directly gives a total running time cost of 
$\Otil_p(m^{\frac{p-2}{(3p-2)}} (m + n)^{\omega} \log(1 / \epsilon))$,
which for large values of $p$,
along with the assumption of $\omega > 2.37$, exceeds $2.70$.

This is more than the running time of about $O(mn^{1.5})$
of algorithms based on inverse maintenance~\cite{Vaidya90,LeeS14,LeeS15}.
In this section we show that the MWU routine from Section~\ref{sec:ResidualAlgo}
can also benefit from fast inverse maintenance.
Our main result is:
\begin{theorem}
\label{thm:FasterMatrixAlgo}
If $\AA$ is an explicitly given, $m$-by-$n$ matrix with polynomially
bounded condition numbers,
and $p \geq 2$
Algorithm~\ref{alg:FasterOracleAlgorithm} as given in
Section~\ref{subsec:WidthReduction} can be implemented to run in total time
\[
\Otil_{p} \left(\left(m + n\right)^{
\max\left\{\omega,
2 + \frac{p - \left(10 - 4 \omega \right)}{3p - 2}
\right\}}\right).
\]
\end{theorem}

\noindent A few remarks about this running time:
the term that dominates depends on the comparison between
$2/3$ and $10 - 4 \omega$, or after manipulation,
the comparison between $\omega$ and $7/3$:
\begin{enumerate}
\item For the current best value of $\omega > 7/3$,
the second term is at most $\omega$, so the total running
time is about $(m + n)^{\omega}$.
\item If $\omega = 2$, then this running time is simply
$(m + n)^{\frac{p - 2}{3p - 2}}$: same as resolving the linear
system at each step.
\item If $\omega \leq 7/3$, then the overhead in the
exponent on the second term is at most
\[
\frac{p - 2/3}{3p - 2} = 1/3,
\]
and this value approaches $\frac{p - 2}{3p - 2}$ as $\omega \rightarrow 2$.
\end{enumerate}

Our algorithm is based on gradually updating the $\rr$ vector.
First, note that $\ww_e^{(i)}$'s, and thus $\rr_{e}^{(i)}$'s are
monotonically increasing.
Secondly, for the $\rr^{(i)}$ that do not double, we can replace
with the original version while forming a factor $2$ preconditioner.
Thus, we only need to update the $\rr^{(i)}$ entries that have
significant increases. This update can be encapsulated by the following result
on computing low rank perturbations to a matrix,
which is a direct consequence of rectangular matrix multiplication
and Woodbury matrix formula.

\begin{lemma}
\label{lem:LowRankUpdate}
Given an $m$-by-$n$ matrix $\AA$,
along with vectors $\rrhat$ and $\rrtil$ that
differ in $k$ entries, as well as the matrix
$\ZZhat = (\AA^{\top} \diag{\rrhat}^{-1} \AA)^{-1}$,
we can construct $(\AA^{\top} \diag{\rrtil}^{-1} \AA)^{-1}$
in $O(k^{\omega - 2} (m + n)^2)$ time.
\end{lemma}

\begin{proof}
Let $S$ denote the entries that differ in $\rrhat$ and $\rrtil$.
Then we have
\[
\AA^{\top} \diag{\rrtil}^{-1} \AA
=
\AA^{\top} \diag{\rrhat}^{-1} \AA
+
\AA_{:,S}^{\top}
\left( \diag{\rrtil_{S}}^{-1} - \diag{\rrhat_{S}}^{-1} \right)
\AA_{S, :}.
\]%
This is a low rank perturbation, so by Woodbury matrix identity we get:
\[
\left( \AA^{\top} \diag{\rrtil}^{-1} \AA
 \right)^{-1}
=
\ZZhat
-
\ZZhat \AA_{:,S}^{\top} \left( 
\left( \diag{\rrtil_{S}}^{-1} - \diag{\rrhat_{S}}^{-1} \right)^{-1}
+ \AA_{S, :} \ZZhat \AA_{:, S}^{\top} \right)^{-1}
\AA_{S, :} \ZZhat,
\]
where we use $\ZZhat^{\top} = \ZZhat$ because
$\AA^{\top} \diag{\rrhat}^{-1} \AA$
is a symmetric matrix.
To explicitly compute this matrix, we need to:
\begin{enumerate}
\item compute the matrix
$\AA_{S,:} \ZZhat$,
\item compute $\AA_{:, S} \ZZhat \AA_{:, S}^{\top}$
\item invert the middle term.
\end{enumerate}
This cost is dominated by the first term, which can be viewed
as multiplying $\lceil n / k \rceil$ pairs of $k \times n$
and $n \times k$ matrices.
Each such multiplication takes time $k^{\omega - 1}n$,
for a total cost of $O(k^{\omega - 2} n^2)$.
The other terms all involve matrices with dimension at most $k \times n$,
and are thus lower order terms.
\end{proof}

Note that the running time of Lemma~\ref{lem:LowRankUpdate}
favours `batching' a large number of modified edges to insert.
To this end, we show that it suffices to have an inverse
that only approximates some entries of $\rr^{(i)}$.
To do so, we first need to introduce our notions of approximations:
\begin{definition}
\label{def:Approx}
We use $a \approx_{c} b$ for positive numbers $a$ and $b$
iff $c^{-1} a \leq b \leq c \cdot b$, and for vectors
and for vectors $\aa$ and $\bb$ we use $\aa \approx_{c} \bb$
to denote $\aa_{i} \approx_{c} \bb_{i}$ entry-wise.
\end{definition}
Since we are only updating $k$ resistances that have a constant factor increase and 
using a constant factor preconditioning for the others, we need the following result 
on preconditioned iterative methods for solving systems of linear equations. %
\begin{lemma}
\label{lem:LOLWhatError}
If $\rr$ and $\rrhat$ are vectors such that
$\rr \approx_{\Otil(1)} \rrhat$, and we're given the matrix
$\ZZhat = (\AA^{\top} \diag{\rrhat}^{-1} \AA)^{-1}$ explicitly,
then we can solve a system
of linear equations involving $\AA^{\top} \diag{\rr}^{-1} \AA$
to $1/\poly(n)$ accuracy in $\Otil(n^2)$ time.
\end{lemma}
As the resistances we provide to $\textsc{Oracle}$
are in the range $[1, O_{p}(m)]$,
we get that each $\rrhat_{e}$ only needs to be updated
$O(\log{m})$ times, instead of after each iteration.
However, it's insufficient to use this bound in the worst-case
manner: if there are $m^{1/3}$ iterations, each of which
doubles the resistances on $m^{2/3}$ edges, then the total
cost as given by Lemma~\ref{lem:LowRankUpdate} becomes
\[
\left(m + n \right)^{2 + \frac{1}{3} + \frac{2}{3} \cdot \left( \omega - 2 \right)},
\]
which is about $(m + n)^{2.58}$.

We get an even better bound by using our analysis from
Section~\ref{sec:ResidualAlgo} to show that for iteration/edge
combinations $i$ and $e$, the (relative) update to $\rr_{e}^{(i)}$ is small.
Such small changes also imply that we can wait on such updates.
For simplicity, suppose we only increment the resistances
by factors of $\frac{1}{L},$ then it takes $\Theta(L)$ such increments until the edge's
resistance has deviated by a constant factor.
Furthermore, we can wait for another $\Theta(L)$ iterations before
having to reflect this change in $\rrhat$: the total relative
increases in these iterations is also at most $O(1)$. Formalizing this process leads to a lazy-update routine
that tracks the increments of different sizes separately.
Its Pseudocode is in Algorithms~~\ref{alg:InverseInit}~and~\ref{alg:InverseMaintenance}.

We will call the initialization routine $\textsc{InverseInit}$
at the first iteration,
and subsequently call {\textsc{UpdateInverse}} %
upon generating a new
set of resistances in the call to Algorithm~\ref{alg:oracle},
\textsc{Oracle}.
This is in turn called from Line~\ref{algline:InvokeOracle}
of Algorithm~\ref{alg:FasterOracleAlgorithm}.
As a result, we will assume access to all variables
of these routines.
Furthermore, our routines keeps the following global variables:
\begin{enumerate}
\item $\rrhat$: resistances from the last time
we updated each entry.
\item $counter(\eta)_{e}$: for each entry, track the number of times that it
changed (relative to $\rrhat$) by a factor of about $2^{-\eta}$
since the previous update.
\item $\ZZhat$, an inverse of the matrix given by $\AA^{\top} \diag{\rrhat}^{-1} \AA$.
\end{enumerate}

\begin{algorithm}[H]
\caption{Inverse Maintenance Initialization}
\label{alg:InverseInit}
 \begin{algorithmic}[1]
\Procedure{InverseInit}{}
\State Set $\rrhat \leftarrow \rr^{(0)}$.
\State Set $counter(\eta)_e \leftarrow 0$ for all $0 \leq \eta \leq \log(m)$ and
$e$. %
\State Set $\ZZ \leftarrow (\AA^{\top} \diag{\rrhat}^{-1} \AA)^{-1}$ by explicitly
inverting the matrix. 
\EndProcedure
\end{algorithmic}
\end{algorithm}

\begin{algorithm}[H]
\caption{Inverse Maintenance Procedure}
\label{alg:InverseMaintenance}
 \begin{algorithmic}[1]
\Procedure{UpdateInverse}{}
\For{all entries $e$}
\State Find the least non-negative integer $\eta$ such that
\[
\frac{1}{2^{\eta}} 
\leq
\frac{\rr^{\left(i\right)}_{e} - \rr^{\left(i - 1\right)}_{e}}{\rrhat_{e}}.
\]
\State Increment $counter(\eta)_{e}$.%
\EndFor
\State $E_{changed}
\leftarrow
\union_{\eta: i \pmod {2^{\eta}} \equiv 0} \{ e: counter(\eta)_e \geq 2^{\eta}\}$ %
\label{lem:EmptyBucket}
\State $\rrtil \leftarrow \rrhat$
\For{all $e \in E_{changed}$}
\State $\rrtil_{e} \leftarrow \rr_{e} ^{\left( i \right)}$. %
\State Set $counter(\eta)_{e} \leftarrow 0$ for all $\eta$. %
\EndFor
\State $\ZZhat \leftarrow \textsc{LowRankUpdate}( \AA, \ZZ, \rrhat, \rrtil)$.
\State $\rrhat \leftarrow \rrtil$.
\EndProcedure 
 \end{algorithmic}
\end{algorithm}

We first verify that the maintained inverse is always a 
good preconditioner to the actual matrix, $\AA^{\top} \diag{\rr^{(i)}} \AA$.
\begin{lemma}
\label{lem:GoodApprox}
After each call to $\textsc{UpdateInverse}$, the vector $\rrhat$
satisfies
\[
\rrhat
\approx_{\Otil \left( 1 \right)}
\rr^{\left( i \right)}.
\]
\end{lemma}

\begin{proof}
  First, observe that any change in resistance exceeding $1$ is
  reflected immediately %
  Otherwise,
  every time we update $counter(j)_e$, $\rr_e$ can only increase
  additively by at most
\[
2^{-j + 1} \rrhat_e.
\]%
Once $counter(j)_e$ exceeds $2^{j}$, $e$ will be added to 
$E_{changed}$ after at most $2^{j}$ steps. So when we start from $\rrhat_e$,  $e$ is added to 
$E_{changed}$ after  $counter(j)_e \leq 
2^j + 2^j = 2^{j+1}$ iterations. The maximum possible increase in resistance due to the bucket $j$ is,
\[
2^{-j + 1} \rrhat_{e}
\cdot
2^{j + 1}
= 4 \rrhat_{e}.
\] %
Since there are only at most $m^{1/3}$ iterations,
the contributions of buckets with $j > \log{m}$ %
are negligible. Now the change in resistance is influenced by all buckets $j$, 
each contributing at most $4\rrhat_{e}$ increase. The total change is at most $4 \rrhat_{e}  \log m$ since there are at most $\log m$ buckets.
We therefore have
\[
\rrhat_e
\leq
\rr^{\left(i\right)}_e
\leq
5 \rrhat_e \log m .
\]
for every $i$.
\end{proof}

It remains to bound the number and sizes of calls
made to Lemma~\ref{lem:LowRankUpdate}.
For this we define variables
\[
k\left(\eta\right)^{\left(i \right)}
\]
to denote the number of edges added to $E_{changed}$
at iteration $i$ due to the value of $counter(\eta)_e$.
Note that $k(\eta)^{(i)}$ is non-zero only if
$i \equiv 0 \pmod{2^{\eta}}$, and
\[
\abs{E_{changed}^{\left(i\right)}}
\leq
\sum_{\eta} k\left(\eta\right)^{\left(i \right)}.
\]
The following lemma gives a lower bound on the relative change of energy across one update of resistances.
\begin{restatable}{lemma}{ChangeEnergy}
  \label{lem:ChangeEnergy}
Assuming the program~\eqref{eq:oracleprog} is feasible,
let $\Delta$ be be a solution to the optimization
problem~\eqref{eq:oracleprog} with weights $\rr$.
Suppose we increase the resistance on each entry to $\rr'$
Then,
\[
\frac{\energy{\rr'} - \energy{\rr}}{\energy{\rr}}
\ge
\Omega_p \left(m^{-(p-2)/p} \sum_e \left(\Delta_e \right)^2 \min \left\{1, \frac{\rr'_e - \rr_e}{\rr_e}\right\}\right).
\]
\end{restatable}
\begin{proof}
  Lemma \ref{lem:ckmst:res-increase} gives us that,
  \begin{align*}
    \frac{\energy{\rr'}}{\energy{\rr}} - 1
    & \geq \exp \left( \frac{\sum_{e}\min \left\{1, \frac{\rr'_e - \rr_e}{\rr_e}\right\}\rr_e\Delta_e^2 }{2\energy{\rr}} \right) - 1\\
    &\geq \left( \frac{\sum_{e}\min \left\{1, \frac{\rr'_e - \rr_e}{\rr_e}\right\}\rr_e\Delta_e^2}{2\energy{\rr}} \right)
  \end{align*}
  Since $\rr_e \geq 1$ and $\energy{\rr} \leq O_p(m^{\frac{p-2}{p}})$,
\[
  \frac{\energy{\rr'}- \energy{\rr}}{\energy{\rr}} \geq \Omega_p
  \left(m^{-(p-2)/p} \sum_e \left(\Delta_e \right)^2 \min \left\{1,
      \frac{\rr'_e - \rr_e}{\rr_e}\right\}\right)
\]
which gives our result.
\end{proof}

We divide our analysis into 2 cases, when the relative change in resistance is at least $1$ and  when the relative change in resistance is at most $1$.  To begin with, let us first look at the following lemma that relates the change in weights to the relative change in resistance. The proof is in the Appendix.
\begin{restatable}{lemma}{ResistanceToFlow}
\label{lem:ResistanceChangeToFlowValue}
Consider a flow step from Line~\ref{algline:LowWidth}
of Algorithm~\ref{alg:FasterOracleAlgorithm}.
We have
\[
\frac{\rr_e^{\left(i + 1\right)}
  - \rr_e^{\left(i\right)}}{\rr_e^{\left(i \right)}}
\leq  \left( 1 + \alpha \abs{\Delta_e} \right)^{p - 2} -1
\]
where $\Delta$ is the $\ell_2$ minimizer solution produced by the oracle.
\end{restatable}

Let us now see what happens when the relative change in resistance is at least $1$.

\begin{lemma}
\label{lem:CountHighWidth}
Throughout the course of a run of Algorithm~\ref{alg:FasterOracleAlgorithm},
the number of edges added to $E_{changed}$ due to  relative resistance increase of at least $1$, 
\[
\sum_{1 \leq i \leq T}
k\left( 0 \right)^{\left( i \right)}
\leq
\Otil_{P}\left(m^{\frac{p + 2}{3p - 2}} \right).
\]
\end{lemma}

\begin{proof}
From Lemma \ref{lem:ChangeEnergy}, we know that the relative change in energy over one iteration is at least,
\[
\Omega_p \left(m^{-(p-2)/p} \sum_e \left(\Delta_e \right)^2 \min \left\{1,\frac{\rr^{(i+1)}_e - \rr^{(i)}_e}{\rr^{(i)}_e}\right\}\right).
\]
Over all iterations, the relative change in energy is at least, 
\[
\Omega_p \left(m^{-(p-2)/p} \sum_i \sum_e \left(\Delta_e \right)^2 \min \left\{1, \frac{\rr^{(i+1)}_e - \rr^{(i)}_e}{\rr^{(i)}_e}\right\}\right) 
\]
 which is upper bounded by $O(\log m)$. When iteration $i$ is a width reduction step, the relative resistance change is always at least $1$. In this case $\abs{\Delta_e} \geq \rho$. When we have a flow step, Lemma \ref{lem:ResistanceChangeToFlowValue} implies that when the relative change in resistance is at least $1$ then,
\[
\alpha \abs{\Delta_e} \geq \Omega_p(1).
\]
This gives, $\abs{\Delta_e} \ge \Omega_p(\alpha^{-1})$. Using this bound on $\abs{\Delta_e}$ is sufficient since $\rho >\Omega_p(\alpha^{-1})$ and both kinds of iterations are accounted for. The total relative change in energy can now be bounded.
\begin{align*}
& \Omega_p \left(m^{-(p-2)/p} \alpha^{-2} \sum_{i}  \sum_e \one_{\left[\frac{\rr^{(i+1)}_e - \rr^{(i)}_e}{\rr^{(i)}_e} \ge 1 \right]} \right) \le \Otil_p(1)\\
& \Leftrightarrow  \Omega_p \left(m^{-(p-2)/p} \alpha^{-2} \sum_{i} k \left( 0 \right)^{\left( i \right)} \right) \le \Otil_p(1)\\
& \Leftrightarrow  \sum_{i} k \left( 0 \right)^{\left( i \right)} \le \Otil_p(m^{(p-2)/p} \alpha^{2}).
\end{align*}
The Lemma follows by substituting $\alpha=  \tilde{\Theta}_p\left(m^{-\frac{p^2-5p+2}{p(3p-2)}}\right)$ in the above equation.
\end{proof}

\begin{lemma}
\label{lem:CountLowWidth}
Throughout the course of a run of Algorithm~\ref{alg:FasterOracleAlgorithm},
the number of edges added to $E_{changed}$ due to  relative resistance increase between $2^{-\eta}$ and $2^{-\eta+1}$, 
\[
\sum_{1 \leq i \leq T}
k\left( \eta \right)^{\left( i \right)}
\leq
\begin{cases}
0 & \text{if $2^{\eta} \geq T$},\\
\Otil_{p}\left(m^{\frac{p + 2}{3p - 2}} 2^{2\eta}\right) & \text{otherwise}.
\end{cases}
\]
\end{lemma}
\begin{proof}
From Lemma \ref{lem:ChangeEnergy}, the total relative change in energy is at least, 
\[
\Omega_p \left(m^{-(p-2)/p} \sum_i \sum_e \left(\Delta_e \right)^2 \left(\frac{\rr^{(i+1)}_e - \rr^{(i)}_e}{\rr^{(i)}_e}\right)\right).
\]
We know that $\frac{\rr^{(i+1)}_e - \rr^{(i)}_e}{\rr^{(i)}_e} \geq 2^{-\eta}$. Using Lemma \ref{lem:ResistanceChangeToFlowValue}, we have,
\[
(1+\alpha \abs{\Delta_e})^{p-2} - 1 \geq 2^{-\eta}. 
\]
We can bound $(1+\alpha \abs{\Delta_e})^{p-2} - 1$ as,
\[
(1+\alpha \abs{\Delta_e})^{p-2} - 1 \leq 
\begin{cases}
\alpha \abs{\Delta_e} & \text{ when $\alpha \abs{\Delta_e} \leq 1$ or $p-2 \leq 1$} \\
O_p\left( \left(\alpha \abs{\Delta_e})^{p-2}\right)\right) & \text{ otherwise. }
\end{cases}
\]
Now, in the second case, when $\alpha \abs{\Delta_e} \geq 1$ and $p-2>1$, 
\[
\left(\alpha \abs{\Delta_e}\right)^{p-2} \geq 2^{-\eta} \Rightarrow \alpha \abs{\Delta_e} \geq \left(\frac{1}{2^{\eta}}\right)^{1/(p-2)} \geq 2^{-\eta}
\]
For both cases we get, 
\[
\alpha \abs{\Delta_e} \geq \Omega_p \left(2^{-\eta}\right).
\]
Using the above bound and the fact that the total relative change in energy is at most $\Otil_p(1)$, gives,
\begin{align*}
&\Omega_p \left(m^{-(p-2)/p} \sum_i \sum_e \left(\Delta_e \right)^2 \left(\frac{\rr^{(i+1)}_e - \rr^{(i)}_e}{\rr^{(i)}_e}\right)\right) \leq \Otil_p(1)\\
\Rightarrow & \Omega_p \left(m^{-(p-2)/p} \sum_i \sum_e \left(\alpha^{-1} 2^{-\eta}\right)^2 \cdot \left(2^{-\eta}\one_{2^{-\eta+1} \geq \frac{\rr^{(i+1)}_e - \rr^{(i)}_e}{\rr^{(i)}_e} \geq 2^{-\eta}}\right) \right)  \leq \Otil_p(1)\\
\Rightarrow & \Omega_p \left(m^{-(p-2)/p} \alpha^{-2} 2^{-3\eta} \sum_i  2^{\eta}  k\left( \eta \right)^{\left( i \right)}\right) \leq \Otil_p(1)\\
\Rightarrow & \sum_i k\left( \eta \right)^{\left( i \right)} \leq \Otil_p \left( m^{(p-2)/p} \alpha^{2} 2^{2\eta}  \right)
\end{align*}
The Lemma follows substituting $\alpha=  \tilde{\Theta}_p\left(m^{-\frac{p^2-5p+2}{p(3p-2)}}\right)$ in the above equation.
\end{proof}

We can now use the concavity of $f(z) = z^{\omega - 2}$ to upper bound
the contribution of these terms.
\begin{corollary}
\label{lem:TotalCostWidth}
Let $k(\eta)^{(i)}$ be as defined. Over all iterations we have,
\[
\sum_{i} \left( k\left( 0 \right)^{\left( i \right)} \right) ^{\omega - 2}
\leq
\Otil_{p} \left( m^{\frac{p - \left(10 - 4 \omega\right) }{3p - 2}} \right)
\]
and for every $\eta$,
\[
\sum_i^T \left( k\left( \eta \right)^{\left( i \right)} \right) ^{\omega - 2} \leq 
\begin{cases}
0 & \text{if $2^{\eta} \geq T$},\\
\Otil_{p} \left( m^{\frac{p - 2 + 4 \left( \omega - 2 \right)}{3p - 2}} \cdot 2^{\eta \left( 3 \omega - 7 \right) } \right) & \text{otherwise}.
\end{cases}
\]
\end{corollary}

\begin{proof}
Due to the concavity of the $\omega - 2 \approx 0.3727 < 1$ power,
this total is maximized when it's equally distributed over all iterations. In the first sum, the number of terms is equal to the number of iterations, i.e.,  
$\Otil_{p}(m^{\frac{p - 2}{3p - 2}})$. In the second sum the number of terms is $\Otil_{p}(m^{\frac{p - 2}{3p - 2}}) 2^{-\eta}$.
Distributing the sum equally over the above numbers give, 
\[
\sum_i^T \left( k\left( 0 \right)^{\left( i \right)} \right) ^{\omega - 2} \leq \left( \Otil_{p} \left( m^{\frac{p+2}{3p-2} - \frac{p-2}{3p - 2}} \right) \right)^{\omega - 2}
\cdot \Otil_{p} \left( m^{\frac{p - 2}{3p - 2}} \right)
=
\Otil_{p} \left( m^{\frac{p - 2 + 4 \left( \omega - 2 \right)}{3p - 2}} \right)
\leq
\Otil_{p} \left( m^{\frac{p - \left(10 - 4 \omega\right) }{3p - 2}} \right)
\]
and 
\begin{align*}
\sum_i^T \left( k\left( \eta \right)^{\left( i \right)} \right) ^{\omega - 2} &\leq \Otil\left( m^{\frac{p - 2}{3p - 2}} 2^{-\eta}\right) \cdot \Otil_{p}\left( \frac{m^{\frac{p + 2}{3p -2}} 2^{2\eta}} {m^{\frac{p - 2}{3p - 2}} 2^{-\eta}} \right)^{\omega - 2}\\
&= \Otil_{p}\left(m^{\frac{p - 2 + 4 \left( \omega - 2 \right) }{3p - 2}}2^{-\eta}\cdot 2^{3\eta (\omega - 2 )}\right)\\
&= \Otil_{p} \left(m^{\frac{p - 2 + 4 \left( \omega - 2 \right)}{3p - 2}} 2^{\eta (3 \omega - 7)}\right).
\end{align*}

\end{proof}

\begin{proof}(of Theorem~\ref{thm:FasterMatrixAlgo}) By
  Lemma~\ref{lem:GoodApprox}, the $\rrhat$ that the inverse being
  maintained corresponds to always satisfy
  $\rrhat \approx_{\Otil(1)} \rr^{(i)}$.  So by the iterative linear
  systems solver method outlined in Lemma~\ref{lem:LOLWhatError},
  we can implement each call to $\textsc{Oracle}$
  (Section~\ref{sec:Oracle})in time $O((n + m)^2)$ in addition to the
  cost of performing inverse maintenance.  This leads to a total cost
  of
\[
\Otil_{p}\left(\left(n + m\right)^{2 + \frac{p - 2}{3p - 2}}\right).
\]
across the $T = \Theta_{p}(m^{\frac{p - 2}{3p - 2}})$ iterations.

The costs of inverse maintenance is dominated by the calls to
the low-rank update procedure outlined in Lemma~\ref{lem:LowRankUpdate}.
Its total cost is bounded by
\[
O\left( \sum_{i} \abs{E_{changed}^{\left( i \right)}}^{\omega - 2}
\left( m + n \right)^2 \right)
=
O\left( 
\left( m + n \right)^2
\sum_{i} \left(
   \sum_{\eta} k\left( \eta \right)^{\left( i \right)}
\right) ^{\omega - 2}
\right).
\]
Because there are only $O(\log{m})$ values of $\eta$,
and each $k(\eta)^{(i)}$ is non-negative, we can bound the total cost by:
\[
\Otil\left( 
\left( m + n \right)^2
\sum_{i} \sum_{\eta}
\left(k\left( \eta \right)^{\left( i \right)}\right)^{\omega - 2}
\right)
\leq
\Otil_{p}\left( 
\left( m + n \right)^2
\sum_{\eta: 2^{\eta} \leq T} 
m^{\frac{p - 2 + 4 \left( \omega - 2 \right)}{3p - 2}}
\cdot 2^{\eta \left( 3 \omega - 7 \right) }
\right),
\]
where the inequality follows from substituting in the result
of Lemma~\ref{lem:TotalCostWidth}.
Depending on the sign of $3 \omega - 7$, this sum is dominated
either at $\eta = 0$ or $\eta = \log{T}$.
Including both terms then gives
\[
\Otil_{p}\left( 
\left( m + n\right)^{2 + \frac{p - 2 + 4 \left( \omega - 2 \right)}{3p - 2}}
+
\left( m + n\right)^{2 + \frac{p - 2 + 4 \left( \omega - 2 \right)
+ \left( p - 2 \right) \left( 3 \omega - 7 \right) }{3p - 2}}
\right),
\]
with the exponent on the trailing term simplifying to $\omega - 2$ to give,
\[
\Otil_{p}\left( 
\left( m + n\right)^{2 + \frac{p - \left(10 - 4 \omega \right)}{3p - 2}}
+
\left( m + n\right)^{\omega}
\right).
\]
\end{proof}

\section{Other Regression Formulations}
\label{sec:Variants}

In this section we discuss how various variants of $\ell_{p}$-norm
regression can be translated into our setting of
\begin{align*}
    \tag{\ref{eq:primal}}
  \min_{\xx} \qquad & \norm{\xx}_{p}\\
              & \AA \xx = \bb.
\end{align*}

As we will address a multitude of problems, we will make generic
numerical assumptions to simplify the derivations.
\begin{tight_enumerate}
\item $m \leq \poly(n)$.
\item All entries in $\AA$ and $\bb$ are at most $\poly(n)$.
Note that this implies that the maximum singular value of $\AA$,
$\sigma_{\max}(\AA)$ and $\norm{\bb}_2$ are at most $\poly(n)$.
\item $\norm{\bb}_{2} \geq 1$.
\item The minimum non-zero singular value of $\AA$,
$\sigma_{\min}(\AA)$ is at least $1 / \poly(n)$.
\end{tight_enumerate}

\noindent Note that these conditions also imply
bounds on the optimum value and the optimum solution $\xx^{*}$: 
Specifically,
\[
OPT \leq \norm{\xx^{*}}_2 \sqrt{m} \leq \norm{\bb}_2 \sigma_{\min}({\AA)}^{-1} \leq \poly\left(n \right),
\]
and
\begin{align*}
OPT \geq \norm{\xx^*}_2 \sqrt{m}^{-1} & \geq  \norm{\bb}_2 \sigma_{\max}\left( \AA \right)^{-1} \cdot m^{-1/2}\\
&  \geq \frac{1}{\poly\left( n \right)}.
\end{align*}

\subsection{Affine transformations within the norm}
\label{subsec:AffineInNorm}
Let $\CC$ be a matrix with the same assumptions as $\AA$ and $\dd$ have assumptions similar to $\bb$. Suppose we are minimizing $\norm{\CC\xx - \dd}_p$ instead of $\norm{\xx}_{p}$, i.e.,
\begin{align*}
\min_{\xx} \qquad & \norm{\CC \xx -\dd}_{p}\\
& \AA \xx = \bb.
\end{align*}

\noindent Note that this can be reduced to the following unconstrained problem,
\[
\min_{\xx} \norm{\CC \xx - \dd}_{p}.
\]

\noindent To see this, first find the null space of $\AA$,
as well as a particular solution $\xx_0$ that satisfies $\AA\xx_0 = \bb$. Let the null space of $\AA$ be generated by the matrix $\VV$. Then the space of solutions can be parameterized as
\[
\xx = \xx_{0} + \VV \yy,
\]
for some vector $\yy$.
Now our objective becomes,
\[
\CC \VV \yy + \left( \CC \xx_{0} - \dd \right),
\]
which can be written as,
\[
\min_{\yy} \norm{\CChat \yy - \ddhat}_{p}.
\]
where $\CChat = \CC\VV$ and $\ddhat = \CC \xx_{0} - \dd$. Observe that $\CChat y$ spans the column space of $\CChat$. Decomposing $\dd$ into a linear combination of an orthonormal basis we could combine the part which is in the span of $\CChat$ with $\CChat\yy$. We can thus replace $\CChat$ in the objective with an orthonormal basis $\UU$ of its column space and replace $\ddhat$ by $\gg$,
a vector orthogonal to all columns of $\UU$. Then any vector
\[
\zz = \UU \yy - \gg,
\]
can equivalently be described by the conditions
\begin{align*}
\zz^{\top} \ddhat_{\perp} & = \norm{\gg}_2^2,\\
\zz^{\top} \vv & = 0, \qquad \forall \vv~s.t.~\UU\vv = 0, \gg^{\top} \vv = 0.
\end{align*}
For the last condition, it suffices to generate an orthonormal basis of the null space of $\UU$.
So the problem can be written as a linear constraint on $\zz$ instead.

\subsection{$1 < p < 2$}
\label{subsec:SmallP}
In case $1 < p < 2,$ we instead solve the dual problem:
\begin{align*}
  \max_{\yy} \qquad & \bb^{\top} \yy\\
              & \norm{\AA^{\top} \yy}_{q} \leq 1,
\end{align*}
for $q  = \frac{p}{p-1}> 2$.
We can rescale the above problem to the equivalent $q$-norm ball-constrained projection problem,
\begin{align*}
  \min_{\yy} \qquad & \norm{\AA^{\top} \yy}_{q}\\
                    & \bb^{\top} \yy = 1,
\end{align*}
where the goal is to check whether the optimum is less than $1$.
This problem is covered by the problem introduced
in Section~\ref{subsec:AffineInNorm} and can thus be solved to high accuracy in the desired time.

It remains is to transform a nearly-optimal solution $\yy$
of this $q$-norm ball-constrained projection problem
to a nearly-optimal solution $\xx$ of the original
subspace $p$-norm minimization problem.
Since both of these problems' solutions are invariant
under scalings to $\AA$ or $\bb$,
we may also assume that the optimum is at most $1$.

\begin{lemma}
\label{lem:GradAlign}
If the optimum of
\begin{align*}
  \max \qquad & \bb^{\top} \yy\\
              & \norm{\AA^{\top} \yy}_{q} \leq 1,
\end{align*}
is at most $1$, and we have some $\yy$ such that
\begin{align*}
  \bb^{\top} \yy & \geq 1 - \delta\\
  \norm{\AA^{\top} \yy}_{q} & = 1,
\end{align*}
then the gradient of $\norm{\AA^{\top} \yy}_q$,
\[
\nabla
=
\AA \sgn{ \AA^{\top} \yy } \left( \AA^{\top} \yy \right)^{q-1},
\]
satisfies
\[
\norm{\nabla - \bb}_2 \leq \delta \poly(n).
\]
\end{lemma}

\begin{proof}
Let $\Delta = \nabla -\bb$ and $p(n)$ be a polynomial such that $\gamma_q \left(\abs{\AA^{\top}\yy}, \AA^{\top}(\nabla - b) \right) \leq p(n)$. By the assumption of $\AA$ and $\bb$ being $\poly(n)$ bounded, the above $\gamma$ function is polynomially bounded. Let $q(n)$ be a polynomial in $n$ such that, $q(n) \geq \sqrt{\frac{4p(n)}{\delta}}$. Suppose, 
\[
\norm{\nabla - \bb}_2^2 \geq \eps >  \delta q(n).
\] 
This gives us, 
\[
\Delta^{\top} \nabla \geq \Delta^{\top} \bb + \eps .
\]

\noindent Now consider the solution
\[
\yyhat\left( \theta \right) \leftarrow \yy - \theta \Delta,
\]
for step size $\theta = \eps/2p(n)$.
Lemma~\ref{lem:LocalApprox} and Lemma \ref{lem:Rescaling} gives
\begin{align*}
\norm{\AA^{\top} \yyhat}_q^q &\leq 1 - \theta \Delta^{\top} \nabla + \gamma_q \left(\abs{\AA^{\top}\yy}, \theta \AA^{\top} \Delta \right)\\
& \leq 1 - \theta \Delta^{\top} \nabla + \theta^2 \gamma_q \left(\abs{\AA^{\top}\yy}, \AA^{\top} \Delta \right)\\
& \leq 1 -  \theta \Delta^{\top} \bb - \theta \epsilon + \theta^2 p(n) \\
& = 1 - \theta  \Delta^{\top} \bb - \frac{ \theta \epsilon}{2}.
\end{align*}
We can scale the solution $\yyhat$ up by a factor of
$1 / (1 - \theta \Delta^{\top} \bb - \theta \epsilon/2)$
to get a solution with objective value
\[
\frac{1 - \delta - \theta \Delta^{\top}\bb}
{1 - \theta \Delta^{\top} \bb - \theta \epsilon/2}.
\]
But by the assumption of $1$ being the optimum, this cannot
exceed $1$, so we get
\[
1 - \delta - \theta \Delta^{\top}\bb
\leq
1 - \theta \Delta^{\top} \bb - \theta \epsilon/2,
\]
or
\[
\theta \epsilon \leq 2 \delta,
\]
which combined with the choice of $\theta $ gives $\eps <  \delta q(n) $ which is a contradiction. So we much have, $ \norm{\nabla - \bb}_2 \leq \delta q(n) \leq \delta \poly(n)$ 
\end{proof}

\noindent This means once $\delta \leq 1 / \poly(n)$,
the solution created from the gradient
\[
\xxhat
\leftarrow 
\sgn{ \AA^{\top} \yy } \left( \AA^{\top} \yy \right)^{q-1},
\]
satisfies
\[
\norm{\AA \xxhat - \bb}_2^2
\leq
\delta \poly\left( n \right).
\]
Also, because $\sigma_{\min}(\AA) \geq 1 / \poly(n)$,
we can create a solution $\xxtil$ from $\xxhat$
by doing a least squares projection on this difference.
This gives:
\[
\AA \xxtil = \bb,
\]
and
\[
\norm{\xx - \xxhat}
\leq
\delta \poly\left(n \right) \sigma_{\min}\left( \AA \right)
\leq
\delta \poly\left( n \right).
\]

\noindent Furthermore, note that because
\[
\left( z^{q - 1} \right)^{p}
= z^{\left( q - 1 \right) \left( p - 1\right) + q - 1}
= z^{q},
\]
we have
\[
\norm{\nabla}_{p}^{p}
=
\norm{\AA \sgn{ \AA^{\top} \yy } \left( \AA^{\top} \yy \right)^{q-1}}_{p}^{p}
= \norm{\AA^{\top} \yy}_{q}^{q}
= 1,
\]
so
\[
\norm{\xx}_{p}
\leq
\norm{\nabla}_{p} + \norm{\xx - \nabla}_{p}
\leq
1 + \delta \poly\left( n \right).
\]

\noindent Thus, for sufficiently small $\delta$, we can get high accuracy
answer to the $p$-norm problem as well.

\section{$p$-Norm Optimization on Graphs}
\label{sec:Graphs}

In this section we discuss the performance of our algorithms on
graphs.  Here instead of invoking general linear algebraic routines,
we instead invoke Laplacian solvers, which provide $1/\poly(n)$
accuracy solutions to Laplacian linear equations in nearly-linear
($\Otil(m)$) time~\cite{SpielmanTengSolver:journal, KoutisMP10,
  KoutisMP11, KelnerOSZ13, CohenKMPPRX14, PengS14, KyngLPSS16,
  KyngS16}, and the current best running time is
$O(m \log^{\nfrac{1}{2}} n \log \nfrac{1}{\eps})$ (up to
${\polyloglog}\ n$ factors)~\cite{CohenKMPPRX14}.

Such matrices can be succinctly described as
\[
\AA^{\top} \diag{\rr}^{-1} \AA
\]
where $\rr$ is the vector of resistances just as
provided in the Oracle from Algorithm~\ref{alg:oracle},
but $\AA$ is the edge-vertex
incidence matrix: with each row corresponding to an edge,
each column corresponding to a vertex, and entries given by:
\[
\AA_{e, v}
= \begin{cases}
1 & \text{if $v$ is the head of $e$,}\\
-1 & \text{if $v$ is the tail of $e$,}\\
0 & \text{otherwise}.
\end{cases}
\]

Throughout this entire section, we will use $\AA$ to
refer to the edge-vertex incidence matrix of a graph.

The main difficult of reducing to Laplacian solvers
is that we can no longer manipulate general matrices.
Specifically, instead of directly working with the
normal matrices as in Section~\ref{subsec:AffineInNorm},
we need to implicitly track the subspaces,
and optimize quadratics on them.
As a result, we need to tailor such reductions towards
the specific problems.

\subsection{$p$-Norm Flows}
\label{sec:Flow}

This is closest to the general regression problem that we study:
\begin{align*}
\min_{\xx} \qquad & \norm{\xx}_{p}\\
& \AA^{\top} \xx = \bb
\end{align*}
except with $\AA$ as an edge vertex incidence matrix.

When $p \geq 2$,
the residual problem then has an extra condition of
\[
\gg^{\top} \ff = \alpha,
\]
which means we need to solve the problem of
\begin{align*}
\min_{\xx} \qquad & \sum_{e} \rr_e \xx_e^2 \\
& \AA^{\top} \xx = \bb\\
& \gg^{\top} \xx = \alpha
\end{align*}
which becomes a solve in the system of linear equations
\begin{align*}
\left[ \AA; \xx \right]^{\top} \diag{\rr}^{-1} \left[ \AA; \xx \right]
=
\left[
\begin{array}{cc}
\AA^{\top} \diag{\rr}^{-1} \AA & \AA^{\top} \diag{\rr}^{-1} \xx\\
\xx^{\top} \diag{\rr}^{-1} \AA & \xx^{\top} \rr^{-1} \xx^{\top}
\end{array}
\right].
\end{align*}
This matrix is a rank $3$ perturbation to the graph
Laplacian $\AA^{\top} \diag{\rr}^{-1} \AA$, and can thus
be solved in $\Otil(m)$ time.
A more detailed analysis of a generalization of this case
can be found in Appendix B of~\cite{DaitchS08}.

When $1 < p < 2$, we invoke the dualization from
Section~\ref{subsec:SmallP} to obtain 
\begin{align*}
\min_{\yy} & \norm{\AA \yy}_{q}\\
& \bb^{\top} \yy = \alpha
\end{align*}
and if we retain the form of $\AA^{\top} \yy$,
but transfer the gradient over to $\yy$, the problem
that we get is:
\begin{align*}
\min_{\yy} \qquad
& \yy^{\top} \AA^{\top} \diag{\rr}^{-1} \AA \yy \\
& \bb^{\top} \yy = \alpha\\
& \gg^{\top} \yy = \beta
\end{align*}
The two additional linear constraints can removed
by writing a variable of $\yy$ as a linear combination of the rest
(as well as $\alpha$).
This then gives an unconstrained minimization problem
on a subset of entries $S$,
\[
\min_{\yy_{S}} \yy_{S}^{\top}
  \LL_{\left[S, S\right]} \yy_{S} + \cc^{\top} \yy_{S}
\]
where $\LL_{[S, S]}$ is a minor of the Laplacian above,
and this solution is obtained by solving for
\[
\yy_{S} \leftarrow \frac{1}{2} \LL_{\left[S, S\right]}^{\dag} \cc.
\]

\subsection{Lipschitz Learning and Graph Labelling}
\label{sec:Label}

This problem asks to label the vertices of a graph,
with a set $T$ fixed to the vector $\ss$,
while minimizing the $p$-norm
difference between neighbours.
It can be written as
\begin{align*}
\min_{\xx: \xx|T = \ss} \qquad & \norm{\AA \xx}_{p}^{p}
\end{align*}
where $\AA$ is the edge-vertex incidence matrix.

In the case of $p \geq 2$,
the residue problem becomes
\begin{align*}
\min_{\xx: \xx|T = \ss} \qquad &
  \xx^{\top} \AA^{\top} \diag{\rr}^{-1} \AA \xx\\
& \gg^{\top} \xx = \beta
\end{align*}

Here the gradient condition can be handled in the same
way as with the voltage problem above: by fixing one
additional entry of $V \setminus T$, and then solving
an unconstrained quadratic minimization problem on the
rest of the variables.

In the case of $1 < p < 2$,
we first write down the problem as an unconstrained
minimization problem on $V \setminus T$:
\begin{align*}
\min_{\xx_{V \setminus T}} \qquad &
  \norm{\AA_{:, V \setminus T} \xx_{V \setminus T} - \AA_{:, T} \ss }_{p}.
\end{align*}
Let $\bb = \AA_{:, T} \ss$ and taking the dual gives:
\begin{align*}
\max \qquad & \bb^{\top} \yy\\
& \norm{\yy}_{q} \leq 1\\
& \left( \AA^{\top} \right)_{V \setminus T, :} \yy = 0
\end{align*}
That is, solving for a small $q$-norm flow that maximizes
the cost against $\bb$, while also having $0$ residues
at the vertices not in $T$.

As $q > 2$, we can now invoke our main algorithm on $\yy$.
Upon binary search, and taking residual problems,
we get $\ell_2$ problems of the form
\begin{align*}
\min \qquad & \sum_{e} \rr_e \yy_e^2\\
& \left( \AA^{\top} \right)_{V \setminus T, :} \yy = 0\\
& \bb^{\top} \yy = \alpha\\
& \gg^{\top} \yy = \beta,
\end{align*}
which is solved by another low rank perturbation
on a minor of the graph Laplacian.

{\small \printbibliography}

\appendix

\section{Missing Proofs}
\label{sec:MissingProofs}

\subsection{Proofs from Section \ref{sec:Prelims}}
\label{sec:ProofsSec3}

\Gamma*
\begin{proof}
\begin{tight_enumerate}
\item We have $p \geq 2$. When $\abs{x} \leq t$, 
\[\gammap(t,x) = \frac{p}{2}t^{p-2}x^2 \geq \frac{p}{2}\abs{x}^p \geq \abs{x}^p.\]
Otherwise, 
\[
\gammap(t,x) = \abs{x}^p + \left(\frac{p}{2}-1\right)t^p \geq \abs{x}^p.
\]
Let $s(x)= \frac{p}{2}t^{p-2}x^2$. At $|x| = t$ we have $\gammap(t,t) = s(t)$. Now, $\gammap'(t,x) = p \abs{x}^{p-2}x \geq p t^{p-2}x = s'(x)$. This means that for $x$ negative, $\gammap$ decreases faster than $s$ and for $x$ positive, $\gammap$ increases faster than $s$. The two functions are equal in the range $-t \leq x \leq t$. Therefore, $\gammap(t,x) \geq s(x)$ for all $x$.

\item
\begin{align*}
\gammap(\lambda t, \lambda x) & =
  \begin{cases}
    \frac{p}{2}(\lambda t)^{p-2}(\lambda x)^2 & \text{if } |x| \leq t,\\
    |\lambda x|^p + (\frac{p}{2} - 1)(\lambda t)^p & \text{otherwise }.
        \end{cases}\\
     & = \begin{cases}
    \lambda^p \frac{p}{2} t^{p-2}x^2 & \text{if } |x| \leq t,\\
   \lambda^p |x|^p + \lambda^p(\frac{p}{2} - 1)t^p & \text{otherwise }.
        \end{cases}\\
        & = \lambda^p \gammap(t,x)
      \end{align*}

\item Taking the derivative of $\gammap(t,x)$ with respect to  $x$ gives,

\[
\frac{d}{dx} \gammap = \begin{cases}
    p t^{p-2}x  & \text{if } |x| \leq t,\\
     p\abs{x}^{p-1} \cdot sign(x)  & \text{otherwise }.
  \end{cases}
\]
The statement clearly follows.
\end{tight_enumerate}
\end{proof}

\Rescaling*
\begin{proof}

\begin{align*}
x \frac{\gammap'(t,x)}{\gammap(t,x)} & = 
\begin{cases}
    2 & \text{if } |x| \leq t,\\
   \frac{p \abs{x}^{p}}{\abs{x}^p + \left(\frac{p}{2} - 1\right)t^p} & \text{otherwise }.
        \end{cases}
\end{align*}
Now, when $\abs{x} \geq t$, we have the following. When $p \leq 2$,
\[
 p = \frac{p \abs{x}^{p}}{\abs{x}^p} \leq \frac{p \abs{x}^{p}}{\abs{x}^p + \left(\frac{p}{2} - 1\right)t^p} \leq  \frac{p \abs{x}^{p}}{ \frac{p}{2}\abs{x}^p} = 2
\]
and when $p \geq 2$,
\[
 2 = \frac{p \abs{x}^{p}}{\frac{p}{2}\abs{x}^p}  \leq  \frac{p \abs{x}^{p}}{\abs{x}^p + \left(\frac{p}{2} - 1\right)t^p} \leq \frac{p \abs{x}^{p}}{\abs{x}^p} = \frac{p}{x}  \]
The above computations imply that,
\[ {\min\{2,p\}} \leq x \frac{\gammap'(t,x)}{\gammap(t,x)} \leq \max \{2,p\}.\]
Let $\lambda \geq 1$ and $x \geq 0$. Integrating both sides of the right inequality gives,
\begin{align*}
&\int_x^{\lambda x} \frac{\gammap'(t,x)}{\gammap(t,x)} dx  \leq \int_x^{\lambda x} \frac{\max \{2,p\}}{x} dx\\
\Leftrightarrow & \log \left(\frac{\gammap(t,\lambda x)}{\gammap(t,x)} \right)  \leq \log \left(\lambda \right)^{\max \{2,p\}}\\
\Leftrightarrow & \gammap(t,\lambda x)  \leq  \lambda^{\max \{2,p\}} \gammap(t,x).
\end{align*}
Integrating both sides of the left inequality from $x$ to $\lambda x$ gives the required left inequality.
Now, let $\lambda \leq 1$. Integrating both sides of the left inequality gives,
\begin{align*}
&\int_{\lambda x}^{x} \frac{\gammap'(t,x)}{\gammap(t,x)} dx  \geq \int_{\lambda  x}^{x} \frac{\min \{2,p\}}{x} dx\\
\Leftrightarrow & \log \left(\frac{\gammap(t,x)}{\gammap(t,\lambda x)} \right)  \geq \log \left(\frac{1}{\lambda} \right)^{\min \{2,p\}}\\
\Leftrightarrow & \gammap(t,\lambda x)  \leq  \lambda^{\min \{2,p\}} \gammap(t,x).
\end{align*}
Similar to the previous case, integrating both sides of the right inequality from $\lambda x$ to $x$ gives the required left inequality. When $x\leq 0$, the direction of the inequality changes but it gets reversed again after putting limits, since we integrate from $\lambda x$ to $x$ when $\lambda \geq 1$ and $x$ to $\lambda x$ when $\lambda \leq 1$.
We thus have, 
\[
 \min\{\lambda^2,\lambda^p\} \gammap(t,\Delta)\leq  \gammap(t,\lambda x)  \leq  \max\{\lambda^2,\lambda^p\} \gammap(t,x)
\]
\end{proof}

\FirstOrderAdditive*
\begin{proof}
  Since $\gammap(t,x) = \gammap(t, \abs{x}),$ and $\gammap(t, x)$ is
  increasing in $x,$ it suffices to prove the claim for
  $x, \Delta \ge 0.$ We have,
  \begin{align*}
    \gammap'(t, x+z) - \gammap'(t, x)
    & = p\max\{t,\  x+z\}^{p-2}(x+z) - p\max\{t,\ \abs{x}\}^{p-2}x \\
    & = p\max\{t^{p-2}(x+z) -  \max\{t,\ \abs{x}\}^{p-2}x,\ \\
    & \qquad\qquad (x+z)^{p-1} -
      \max\{t,\ \abs{x}\}^{p-2}x \} \\
    & \le p\max\{t^{p-2}(x+z) -  t^{p-2}x,\  (x+z)^{p-1} -  x^{p-1} \}
    && \text{(Since $p \ge 2$)}\\
    & \le p\max\{t^{p-2}z,\  (p-1)(x+z)^{p-2}z\}
    && \text{(Using Rolle's   theorem)} \\
    & \le p \max\{t^{p-2}z,\  p(2x)^{p-2}z,\  p(2z)^{p-2}z\}\\
    & \le p^{2}2^{p-2} \max\{t,x,\Delta\}^{p-2} z
    && \text{(Since $z \le \Delta$)}
  \end{align*}
Integrating over $z \in [0,\Delta],$ we get,
\begin{align*}
  \gammap(x+\Delta) - \gammap(x) - \Delta \gammap'(x)
  & \le p^{2} 2^{p-3} \max\{t, \abs{x}, \abs{\Delta}\}^{p-2} \Delta^{2}.
\end{align*}
\end{proof}

\subsection{Proofs from Section \ref{sec:MainAlgo}}
\label{sec:ProofsSec4}

\BoundInitial*

\begin{proof} We first show the following two lemmas.
\begin{lemma}\label{smallalpha}
For $|\alpha| \leq1 $ and $p \geq 1$, 
\begin{equation*}
1 + \alpha p + \frac{(p-1)}{4} \alpha^2 \leq (1+\alpha)^p \leq 1 + \alpha p + p2^{p-1} \alpha^2.
\end{equation*}
\end{lemma}
\begin{proof}
Let us first show the left inequality, i.e. $1 + \alpha p + \frac{p-1}{4} \alpha^2 \leq (1+\alpha)^p $. Define the following function,
\begin{equation*}
h(\alpha) = (1+\alpha)^p - 1 - \alpha p - \frac{p-1}{4} \alpha^2.
\end{equation*}
When $\alpha = 1,-1$, $h(\alpha) \geq 0$. The derivative of $h$ with respect to $\alpha$ is, $h'(\alpha) = p(1+ \alpha)^{p-1} - p- \frac{(p-1)}{2} \alpha $. \\
When $p\geq 2$ and $-1 <\alpha <1$, 
\begin{align*}
&\left( (1+ \alpha)^{p-2} - 1 \right) sign(\alpha)  \geq 0 \\
&\Rightarrow \left( (1+ \alpha)^{p-1} - (1+\alpha) \right) sign(\alpha) \geq 0\\
& \Rightarrow \left( p(1+ \alpha)^{p-1} - p - p\alpha \right) sign(\alpha) \geq 0\\
&\Rightarrow \left( p(1+ \alpha)^{p-1} - p - \frac{(p-1)}{2}\alpha \right) sign(\alpha) \geq 0
\end{align*}
For the last inequality, note that when the product is positive, either both terms are positive or both terms are negative. When both terms are positive, subtracting $(p-1)/2$ instead of $p$ gives a larger positive quantity. When both terms are negative then subtracting $(p-1)/2$ instead of $p$ gives only a smaller quantity, so the inequality holds. This shows that $h'(\alpha) sign(\alpha) \geq 0$, which means minimum of $h$ is at $h(0)=0$. Next let us see what happens when $p \leq 2$ and $\abs{\alpha} <1$. 

\[
h''(\alpha) = p(p-1)(1+\alpha)^{p-2} - \frac{p-1}{2} = (p-1) \left(\frac{p}{(1+\alpha)^{2-p}} - \frac{1}{2}\right) \geq 0
\]
This implies that $h'(\alpha)$ is an increasing function of $\alpha$ and $\alpha_0$ for which $h'(\alpha_0) = 0$ is where $h$ attains its minimum value. The only point where $h'$ is 0 is $\alpha_0 = 0$. This implies $h(\alpha) \geq h(0) = 0$. This concludes the proof of the left inequality. For the right inequality, define:
\begin{equation*}
s(\alpha) = 1 + \alpha p + p2^{p-1} \alpha^2 - (1+\alpha)^p.
\end{equation*}
Note that $s(0) = 0$ and $s(1),s(-1) \geq 0$. We have, 
\[
s'(\alpha) = p + p2^{p} \alpha - p(1+\alpha)^{p-1}
\]
Using the mean value theorem for $p\geq2$ and $\alpha <0$,
\begin{align*} 
(1+\alpha)^{p-1} - 1 &= (p-1) \alpha \cdot (1+z)^{p-2} , z\in (\alpha,0)\\
& \geq  \alpha  2^{p}.
\end{align*} 
This implies that $s'(\alpha) \leq 0$ for negative alpha. When $ 1 > \alpha > 0$, using the convexity of $f(x) = (1+x)^{p-1}$ for $p > 2$, we get,
\[
f(\alpha \cdot 1 + (1-\alpha) \cdot 0) \leq \alpha f(1) + (1-\alpha) f(0)
\]
which gives us 
\[
(1+\alpha)^{p-1} \leq \alpha 2^{p-1} + 1.
\]
This implies, $s'(\alpha)\geq 0$ for positive $\alpha$. The function $s$ is thus increasing for positive $\alpha$ and decreasing for negative $\alpha$, so it attains the minimum at $0$ which is $s(0) = 0$ giving us $s(\alpha) \geq 0$.  We now look at the case $p\leq2$. We have 
\[ 
(1+\alpha)^{p-1}sign(\alpha) \leq (1+\alpha)sign(\alpha).
\]
Using this, we get, $s'(\alpha) sign(\alpha) \geq p|\alpha|  (2^{p} - 1) \geq 0$ which says $s'(\alpha)$ is positive for $\alpha$ positive and negative for $\alpha$ negative. Thus the minima of $s$ is at 0 which is $0$. So $s(\alpha) \geq0$ in this range too.

\end{proof}

\begin{lemma}\label{beta}
For $\beta \geq 1$ and $p \ge 1$, $(\beta -1 )^{p-1} +1 \geq  \frac{1}{2^p} \beta^{p-1}$.
\end{lemma}
\begin{proof}
$(\beta -1 ) \geq \frac{\beta}{2}$ for $\beta \geq 2$. So the claim clearly holds for $\beta \geq 2$ since $(\beta -1 )^{p-1} \geq \left(\frac{\beta}{2}\right)^{p-1}$. When $1 \leq \beta \leq 2$, $1 \geq \frac{\beta}{2}$, so the claim holds since, $1 \geq \left(\frac{\beta}{2}\right)^{p-1}$
\end{proof}

\noindent We now prove the theorem.

  Let $\Delta = \alpha x$. The term $g \Delta = p |x|^{p-1} sign(x) \cdot \alpha x = \alpha p |x|^{p-1}|x| = \alpha p |x|^p$. Let us first look at the case when $|\alpha| \leq 1$. We want to show, 
\begin{align*}
& |x|^p + \alpha p |x|^p + c\frac{p}{2}|x|^{p-2}|\alpha x|^2  \leq |x+\alpha x|^p \leq |x|^p + \alpha p |x|^p + C\frac{p}{2}|x|^{p-2}|\alpha x|^2 \\
& \iff (1+ \alpha p) + c\frac{p}{2}\alpha^2  \leq (1+\alpha)^p \leq (1 + \alpha p)+ C\frac{p}{2}\alpha^2.
\end{align*}
This follows from Lemma~\ref{smallalpha} and the facts $\frac{cp}{2} \leq \frac{p-1}{4}$ and $\frac{Cp}{2} \geq p2^{p-1}$ . We next look at the case when $|\alpha| \geq 1$.  Now, $\gamma_{|f|}^p (\Delta) = |\Delta|^p + (\frac{p}{2} - 1)|f|^p$. We need to show

\begin{equation*}
|x|^p (1 + \alpha p)+ \frac{|x|^p(p-1)}{p2^p}( |\alpha|^p + \frac{p}{2} - 1) \leq |x|^p|1 + \alpha|^p \leq  |x|^p (1 + \alpha p)+2^p |x|^p( |\alpha|^p + \frac{p}{2} - 1).
\end{equation*}

When $|x| = 0$ it is trivially true. When $|x| \neq 0$, let 
\begin{equation*}
h(\alpha) = |1+\alpha|^p - (1+\alpha p) - \frac{(p-1)}{p2^p}(|\alpha|^p + \frac{p}{2} - 1).
\end{equation*}
Now, taking the derivative with respect to $\alpha$ we get,
\begin{equation*}
h'(\alpha) = p \left( |1+\alpha|^{p-1}sign(\alpha) - 1 - \frac{(p-1)}{p2^p} |\alpha|^{p-1}sign(\alpha)\right).
\end{equation*}
When $\alpha \geq1$ and $p \geq 2$,
\[
h''(\alpha) \geq p(p-1)(1+\alpha)^{p-2} - p\frac{(p-1)}{2^p} \alpha^{p-2} \geq 0.
\] 
So we have $h'(\alpha) \geq h'(1) \geq 0$. When $p <2$, we use the mean value theorem to get,
\begin{align*}
(1+\alpha)^{p-1} - 1 &= (p-1) \alpha (1+z)^{p-2}, z \in (0,\alpha)\\
& \geq (p-1)\alpha (2\alpha)^{p-2} \\
& \geq \frac{p-1}{2} \alpha^{p-1}
\end{align*}
which implies $h'(\alpha) \geq 0$ in this range as well. When $\alpha \leq -1$ it follows from Lemma~\ref{beta} that $h'(\alpha) \leq0$. So the function $h$ is increasing for $\alpha \geq 1$ and decreasing for $\alpha \leq -1$. The minimum value of $h$ is $min \{h(1), h(-1) \} \geq 0$.  It follows that $h(\alpha) \geq 0$ which gives us the left inequality. The other side requires proving,
\begin{equation*}
|1+\alpha|^p \leq 1+\alpha p + 2^p (|\alpha|^p + \frac{p}{2} -1).
\end{equation*}
Define:
\begin{equation*}
s(\alpha) = 1+\alpha p + 2^p (|\alpha|^p + \frac{p}{2} -1) - |1+\alpha|^p.
\end{equation*}
The derivative $s'(\alpha) = p + \left(p 2^p |\alpha|^{p-1} - p|1+\alpha|^{p-1} \right)sign(\alpha)$ is non negative for $\alpha \geq 1$ and non positive for $\alpha \leq -1$.  The minimum value taken by $s$ is $\min\{s(1),s(-1)\} $ which is non negative. This gives us the right inequality.

\end{proof}

\InitialPoint*
\begin{proof}
Let $\xx^{\star}$ give the \opt. We know that, for any $\xx$,
\[
\pnorm{\xx} \leq \norm{\xx}_2 \leq m^{\nfrac{(p-2)}{2p}} \pnorm{\xx}.
\]
This along with the fact $\smallnorm{\xx^{(0)}}_2 \leq \smallnorm{\xx^{\star}}_2$ gives us,
\[
\smallnorm{\xx^{(0)}}_p^p \leq \smallnorm{\xx^{(0)}}_2^p \leq m^{\nfrac{(p-2)}{2}}\smallnorm{\xx^{\star}}_p^p = m^{\nfrac{(p-2)}{2}}\opt.
\]
\end{proof}

\subsection{Proofs from Section \ref{sec:ResidualAlgo}}

\label{sec:ProofsSec5}

\BoundResidual*
\begin{proof}
Let $\xx^{\star}$ denote the optimum solution of \eqref{eq:primal} and $\xx^{(0)}$ be as defined in Definition \ref{start}. We know that for any $\xx$,
\[
\pnorm{\xx} \leq \norm{\xx}_2 \leq m^{\nfrac{(p-2)}{2p}} \pnorm{\xx}.
\]
This along with the fact $\smallnorm{\xx^{(0)}}_2 \leq \smallnorm{\xx^{\star}}_2$ gives us,
\[
\smallnorm{\xx^{(0)}}_p^p \leq \smallnorm{\xx^{(0)}}_2^p \leq m^{\nfrac{(p-2)}{2}}\smallnorm{\xx^{\star}}_p^p.
\]
Now from Lemma \ref{lem:LambdaBound} we have,
\[
\residual(\Delta) \leq \frac{1}{\lambda} (\smallnorm{\xx^{(0)}}_p^p - \smallnorm{\xx^{\star}}_p^p) \leq \frac{\smallnorm{\xx^{(0)}}_p^p}{\lambda}(1 - m^{\nfrac{-(p-2)}{2}}).
\]
Let us assume $\residual(\Delta) \geq \eps \smallnorm{\xx^{\star}}_p^p \geq \eps m^{\nfrac{-(p-2)}{2}} \smallnorm{\xx^{(0)}}_p^p$. If this is not true we already have an $\eps$ approximate solution to our problem. We thus have the following bound on $\residual$,
\[
\eps m^{\nfrac{-(p-2)}{2}} \smallnorm{\xx^{(0)}}_p^p \leq \residual(\Delta) \leq \frac{\smallnorm{\xx^{(0)}}_p^p}{\lambda}.
\]
This gives us that, 
\[2^{\log \left(\frac{\eps\smallnorm{\xx^{(0)}}_p^p}{m^{\nfrac{(p-2)}{2}}}\right)} \leq  \residual(\Delta) \leq 2^{\log \left(\frac{\smallnorm{\xx^{(0)}}_p^p}{\lambda} \right)}.\]
When $p \leq 2$, following a similar proof and using,
\[
\norm{\xx}_2 \leq \pnorm{\xx} \leq m^{\nfrac{(2-p)}{2p}} \norm{\xx}_2,
\]
we get,
\[2^{\log \left(\frac{\eps\smallnorm{\xx^{(0)}}_p^p}{m^{\nfrac{(2-p)}{2}}}\right)} \leq  \residual(\Delta) \leq 2^{\log \left(\frac{\smallnorm{\xx^{(0)}}_p^p}{\lambda} \right)},\]
thus concluding the proof of the lemma.
thus concluding the proof of the lemma.
\end{proof}

\StepProblemSearch*

\begin{proof}
  Assume that the optimum solution to~\eqref{eq:residual},
  $\Delta^{\star}$ satisfies
  \[
    \alpha(\Delta^\star) = 
  \gg^{\top} \Delta^{\star} - \frac{p-1}{p2^p}\gamma_{p} \left( \tt, \Delta^{\star}
  \right) \in \left[2^{i - 1}, 2^{i}\right),
\]
in addition to $\AA\Delta^{\star} = 0.$ Note that we know that the
objective is strictly positive (as 0 is a feasible solution).  Since
$\gammap \ge 0,$ we must have,
\[
  \gg^T \Delta^{\star} \geq 2^{i-1},
\]

Consider scaling $\Delta^{\star}$ by a factor $\lambda > 0.$ Since
$\Delta^{\star}$ is optimal, we must have
\[
  \frac{d}{d\lambda} \left. (\lambda \gg^{\top} \Delta^{\star} -
    \frac{p-1}{p2^p}\gamma_{p}(\tt,\lambda \Delta^{\star})) \right|_{\lambda = 1} = 0.
\]
Now, from Lemma~\ref{lem:Rescaling}, we know that
\[ \frac{\min\{2,p\}}{\lambda}\gamma_{p} (\tt,\lambda \Delta^{\star}) \le
  \frac{d}{d\lambda} \gamma_{p}(\tt,\lambda \Delta^{\star}) \le
  \frac{\max\{2,p\}}{\lambda}\gamma_{p} (\tt,\lambda \Delta^{\star}).
\]
Thus, we get,
\[ \frac{p-1}{p2^p} \min\{2,p\} \gamma_{p} (\tt,\Delta^{\star})
  \le \gg^{\top} \Delta^{\star}
  \le \frac{p-1}{p2^p} \max\{2,p\} \gamma_{p} (\tt,\Delta^{\star}),
\]
\[\frac{p-1}{p2^p} (\min\{2,p\}-1) \gamma_{p} (\tt,\Delta^{\star})
  \le \gg^{\top} \Delta^{\star} - \frac{p-1}{p2^p} \gamma_{p} (\tt,\Delta^{\star})
  \le \frac{p-1}{p2^p}(\max\{2,p\}-1)\gamma_{p} (\tt,\Delta^{\star}).
\]
Thus,
$\gamma_{p} (\tt,\Delta^{\star}) \le \frac{1}{\frac{p-1}{p2^p}\left(\min\{2,p\}-1\right)}2^{i},$
and hence $\gg^{\top} \Delta^{\star} \le
\frac{\max\{2,p\}}{\min\{2,p\} -1}2^{i} =
\max\left\{ p,\nfrac{2}{p-1} \right\} 2^{i}
$.

Now consider the vector $\Delta = \lambda \Delta^{\star},$ where
$ \lambda = \frac{2^{i-1}}{\gg^{\top} \Delta^{\star}}.$ Note that
$\lambda \in \left[ \min\left\{ \nfrac{1}{2p}, \nfrac{(p-1)}{4} \right\}
  , 1 \right].$ We have
\begin{align*}
  \gg^{\top} \Delta
  & = \gg^{\top} \left(  \frac{2^{i-1}}{\gg^{\top}
    \Delta^{\star}}\Delta^{\star} \right) = 2^{i-1} \\
  \gamma_{p}(\tt, \Delta)
  & \le  \gamma_{p}(\tt, \lambda \Delta^{\star})
    \le  \max\{\lambda^{2}, \lambda^{p}\}\gamma_{p}(\tt, \Delta^{\star})
     \le \frac{p2^p}{p-1} 2^{i}.
\end{align*}
Thus, $\Delta$ is a feasible solution to
Program~\eqref{eq:BinarySearchProblems}.
A $\beta$-approximate solution $\Delta(i)$ must be such that,
\begin{align*}
  \AA \Delta(i) & = 0, \\
  \gg^{T} \Delta(i)
  & = 2^{i-1}, \\
  \gamma_{p}\left( \tt, \Delta(i) \right)
  & \leq \beta \frac{p2^p}{p-1}2^{i}.
\end{align*}

Now, we consider $\Delta = \mu \Delta(i)$ for some $\mu \le 1.$ We have, $\AA \Delta = 0,$
and,
\begin{align*}
  \gg^{\top}\Delta -  \frac{p-1}{p2^p} \gamma_{p}\left( \tt, \Delta  \right)
  & =  \mu \gg^{\top}
    \Delta(i) - \frac{p-1}{p2^p} \gamma_{p}\left( \tt, \mu  \Delta(i)  \right) \\
  & \ge  \mu 2^{i-1} -  \max\{\mu^{2},\mu^{p}\} 
    \beta 2^{i}
  && \text{(Using Lemma~\ref{lem:Rescaling})}
\end{align*}
We can pick,
\[
  \mu \leq
  \begin{cases}
    \left( \frac{1 }{2 \beta p} \right)^{\frac{1}{p-1}} & \text{ if
    } p \le 2 \\
    \frac{1}{4\beta} & \text{ if } p \ge 2.
  \end{cases}
\]
In either case, we get,
\[\gg^{\top} \Delta - \frac{p-1}{p2^p}\gammap(\tt, \Delta) \ge \mu \left( 1-\frac{1}{\min\{2,p\}}\right)2^{i-1}\]

Since we assumed that the optimum of Program~\eqref{eq:residual} is at
most $2^{i},$ this implies that $ \mu \Delta(i)$ achieves an
objective value for Program~\eqref{eq:residual} that is within an
$\Omega_{p}\left(\beta^{-\frac{1}{\min\{p,2\} - 1}} \right)$ fraction
of the optimal.
\end{proof}

\Dual*

\begin{proof}
We choose $i$ such that \eqref{eq:BinarySearchProblems} is feasible, i.e., there exists $\Delta$ such that,
\begin{align*}
    \begin{aligned}
      \gamma_{p} \left( \tt , \Delta \right) &
      \le \frac{p2^p}{p-1} 2^{i}, \\
      \gg^T \Delta & = 2^{i-1}, \\
      \AA \Delta & = 0.
    \end{aligned}
  \end{align*}
Scaling both $\tt$ and $\Delta$ to $\tilde{\tt} =\left(\frac{p-1}{p}\right)^{1/p} 2^{-1-i/p }\tt$ and $\tilde{\Delta} = \left(\frac{p-1}{p}\right)^{1/p}2^{-1-i/p}\Delta$ gives us the following.
\begin{align*}
    \begin{aligned}
      \gamma_{p} \left(\tilde{\tt} , \tilde{\Delta} \right) &
      \le  1, \\
      \gg^T \tilde{\Delta} & = \left(\frac{p-1}{p}\right)^{1/p} 2^{i(1-1/p)-2}, \\
      \AA \tilde{\Delta} & = 0.
    \end{aligned}
  \end{align*}
Now, let $\tt' =  \max\{m^{-1/p},\tilde{\tt_e}\}$. We claim that when $p \geq 2$, $\gamma_{p} \left(\tt' , \tilde{\Delta} \right) - \gamma_{p} \left(\tilde{\tt} , \tilde{\Delta} \right) \leq \frac{p}{2} - 1$. To see this, for a single $j$, let us look at the difference $\gamma_{p} \left(\tt'_j , \tilde{\Delta}_j \right) - \gamma_{p} \left(\tilde{\tt}_j , \tilde{\Delta}_j \right)$. If $\tilde{\tt}_j \geq m^{-1/p}$ the difference is $0$. Otherwise from the proof of Lemma 5 of \cite{BubeckCLL18}, 
\[
\gamma_{p} \left(\tt'_j , \tilde{\Delta}_j \right) - \gamma_{p} \left(\tilde{\tt}_j , \tilde{\Delta}_j \right) \leq \gamma_{p} \left(\tt'_j , \tilde{\Delta}_j \right) - \abs{\tilde{\Delta}_j}^p \leq \left( \frac{p}{2} - 1\right) \left(m^{-1/p}\right)^p.
\]
When $p \leq 2$, we claim that $ \gamma_{p} \left(\tilde{\tt} , \tilde{\Delta} \right) - \gamma_{p} \left(\tt' , \tilde{\Delta} \right) \leq 1 -\frac{p}{2}$. Again if $\tilde{\tt}_j \geq m^{-1/p}$ the difference is $0$. Otherwise, 
\[
\gamma_{p} \left(\tilde{\tt}_j , \tilde{\Delta}_j \right) - \gamma_{p} \left(\tt'_j , \tilde{\Delta}_j \right) \leq  \abs{\tilde{\Delta}_j}^p - \gamma_{p} \left(\tt'_j , \tilde{\Delta}_j \right)  \leq \left(1- \frac{p}{2} \right) \left(m^{-1/p}\right)^p.
\]
To see the last inequality, when $\abs{\Delta_j} \leq \tt'_j$, we require, $\abs{\Delta_j}^p - \frac{p}{2}\tt^{p-2}_j \Delta_j^2 \leq \left(1- \frac{p}{2}\right)\tt_j^p$ which is true. When $\abs{\Delta_j} \geq \tt_j$, it directly follows.
Summing over all $j$ gives us our claims. We know that $\gamma_{p} \left(\tilde{\tt} , \tilde{\Delta} \right) \leq 1$. Thus, $\gamma_{p} \left(\tt' , \tilde{\Delta} \right) \leq \frac{p}{2}$. Next we set $\hat{\Delta} = \left(\frac{2}{p}\right)^{1/2}\tilde{\Delta}$. Note that $\max\{\left(\frac{2}{p}\right)^2,\left(\frac{2}{p}\right)^p \} = \left(\frac{2}{p}\right)^2$ for all $p$. Lemma \ref{lem:Rescaling} thus implies,
\[
\gamma_{p} \left(\tt' , \hat{\Delta} \right) \leq \left(\frac{2}{p}\right) \gamma_{p} \left(\tt' , \tilde{\Delta} \right) \leq 1.
\]
Define $\hat{t}_j = \min\{1, \tt'_j\}$. Note that $\gamma_{p} \left(\hat{\tt} , \hat{\Delta} \right) = \gamma_{p} \left(\tt' , \hat{\Delta} \right)$ since $\gamma_{p} \left(\tt' , \hat{\Delta} \right) \leq 1$ and as a result we have $\gamma_{p} \left(\hat{\tt} , \hat{\Delta} \right) \leq 1$. Observe that $\hat{\Delta}$ is a feasible solution of \eqref{eq:ScaledProblem} thus suggesting that for problem  \eqref{eq:ScaledProblem} $\opt \leq 1$. 
Let $\Delta^{\star}$ be a $\kappa$ - approximate solution to \eqref{eq:ScaledProblem}, i.e.,
\[
\gamma_{p} \left(\hat{\tt} , \Delta^{\star} \right) \leq \kappa \cdot \opt \leq \kappa.
\]
When $p\geq 2$, $\gammap$ is an increasing function of $\tt$ giving us,
\[
\gamma_{p} \left(\tilde{\tt} , \Delta^{\star} \right) \leq \gamma_{p} \left(\tt' , \Delta^{\star} \right) = \gamma_{p} \left(\hat{\tt} , \Delta^{\star} \right) \leq \kappa.
\]
When $p\leq 2$, 
\[
\gamma_{p} \left(\tilde{\tt} , \Delta^{\star}\right) \leq \gamma_{p} \left(\tt' , \Delta^{\star} \right)  + 1 - \frac{p}{2} \leq \kappa+1
\]
This gives,
\[
\gamma_{p} \left(\tt , \left(\frac{p}{p-1}\right)^{1/p} 2^{1+i/p} \Delta^{\star} \right) \leq \frac{p2^p}{p-1} 2^i (\kappa+1)
\]
and Lemma \ref{lem:Rescaling} then implies,
\[
\gamma_{p} \left(\tt , \left(\frac{p}{2}\right)^{1/2} \left(\frac{p}{p-1}\right)^{1/p} 2^{1+i/p} \Delta^{\star} \right) \leq \left(\frac{p}{2}\right)^{p/2} \frac{p2^p}{p-1} 2^i (\kappa+1).
\]
Finally, $\Delta =   \left(\frac{p}{2}\right)^{1/2} \left(\frac{p}{p-1}\right)^{1/p} 2^{1+i/p} \Delta^{\star}$ satisfies the constraints of \eqref{eq:BinarySearchProblems} and is a $\Omega_p(\kappa)$ approximate solution.
\end{proof}

\BoundOpt*
\begin{proof}
Using H\"{o}lder's inequality, we have,
\begin{align*}
  \sum_e \ww_e^{p-2} (\Delta^*_e)^2
  & \leq \left(\sum_e ((\Delta^*_e)^2)^{p/2} \right)^{2/p} \left(\sum_e (\abs{\ww_e}^{p-2})^{p/(p-2)} \right)^{(p-2)/p}\\
  & =  \left(\sum_e \abs{\Delta_e^*}^p \right)^{2/p} \left(\sum_e \abs{\ww_e}^{p} \right)^{(p-2)/p}\\
  & \leq \left(\sum_e \abs{\ww_e}^{p} \right)^{\frac{(p-2)}{p}},
    \text{    since $\sum_{e} (\Delta^*_e)^p \leq 1$}.
\end{align*}
\end{proof}

\subsection{Proofs from Section \ref{sec:InvMaintain}}
\label{sec:ProofsSec6}

\ResistanceToFlow*
\begin{proof}
Recall from the setting of resistances
from Line~\ref{algline:resistance} of \textsc{Oracle}
(Algorithm~\ref{alg:oracle}) that
\[
\rr_e^{\left(i \right)}
=
\left(m^{1/p} \tt_e\right)^{p - 2}
+ \left( \ww_e^{\left(i\right)}\right)^{p - 2}.
\]
By Line~\ref{algline:LowWidth} of Algorithm~\ref{alg:FasterOracleAlgorithm},
we have
\[
\ww^{\left( i + 1 \right)}_{e} - \ww^{\left( i \right)}_{e}
=
\alpha \abs{\Delta_e}.
\]
Substituting this in gives
\[
\frac{\rr_e^{\left(i + 1\right)} - \rr_e^{\left(i\right)}}{\rr_e^{\left(i \right)}}
=
\frac{\left( \ww_e^{\left(i\right)} + \alpha \abs{\Delta_e}\right)^{p - 2}
- \left( \ww_e^{\left(i\right)}\right)^{p - 2}}
{\left(m^{1/p} \tt_e\right)^{p - 2}
+ \left( \ww_e^{\left(i\right)}\right)^{p - 2}}.
\]
There are two cases to consider: \sushant{This proof can be made easier}
\begin{tight_enumerate}
\item $\ww_e^{(i)} \geq m^{1/p} \tt_e$.
\[
\frac{\rr_e^{\left(i + 1\right)} - \rr_e^{\left(i\right)}}{\rr_e^{\left(i \right)}}
\leq
\frac{\left( \ww_e^{\left(i\right)} + \alpha \abs{\Delta_e}\right)^{p - 2}
- \left( \ww_e^{\left(i\right)}\right)^{p - 2}}
{\left( \ww_e^{\left(i\right)}\right)^{p - 2}}
\leq
\left( 1 + \frac{\alpha \abs{\Delta_e}}{\ww_e^{\left( i \right)}} \right)^{p - 2}
-1
\leq
\left( 1 + \alpha \abs{\Delta_e} \right)^{p - 2}
-1
\]
where the last inequality utilizes $\ww_e^{(i)} \geq 1$,
which is due to the assumption and $m^{1/p} \tt_e \geq 1$.

\item $\ww_e^{(i)} \leq m^{1/p} \tt_e$, then replacing the
denominator with the $(m^{1/p} \tt_e)^{p - 2}$ term and simplifying gives
\[
\frac{\rr_e^{\left(i + 1\right)} - \rr_e^{\left(i\right)}}{\rr_e^{\left(i \right)}}
\leq
\left( \frac{\ww_e^{\left(i\right)}}{m^{1/p} \tt_e}
  + \frac{\alpha \abs{\Delta_e}}{m^{1/p} \tt_e}\right)^{p - 2}
- \left(  \frac{\ww_e^{\left(i\right)}}{m^{1/p} \tt_e} \right)^{p - 2}.
\]
As the function $(z + \theta)^{p - 2} - z^{p - 2}$ is monotonically
increasing when $\theta, p - 2 \geq 0$, we may replace the
$\frac{\ww_e^{\left(i\right)}}{m^{1/p} \tt_e}$ by its upper of $1$
(given by the assumption) to get
\[
\frac{\rr_e^{\left(i + 1\right)} - \rr_e^{\left(i\right)}}{\rr_e^{\left(i \right)}}
\leq
\left( 1 + \frac{\alpha \abs{\Delta_e}}{m^{1/p} \tt_e}\right)^{p - 2}
- 1
\leq
\left( 1 + \alpha \abs{\Delta_e}\right)^{p - 2}
- 1,
\]
where the last inequality follows from $m^{1/p} \tt_e \geq 1$. 
\end{tight_enumerate}
\end{proof}

\section{Controlling $\Phi$}
\label{sec:ControllingGammaPotential}

\ReduceWidthGammaPotential*
\begin{proof}
  We prove this claim by induction. Initially, $i = k = 0,$
  and $\Phi(0,0) = 0,$ and thus, the claim holds trivially. Assume that the
  claim holds for some $i,k \ge 0.$
We will use $\Phi$ as an abbreviated notation for $\Phi(i,k)$ below.
 
\paragraph{Flow Step.}
For
  brevity, we let $\gammap(\ww)$ denote
  $\gammap(m^{\nfrac{1}{p}} \tt, \ww),$ 
and use $\ww$ to denote $\ww^{(i,k)}$.

If the next step is a \emph{flow} step, 
  \begin{align*}
    \Phi\left( i+1,k \right)
    =
    & \gammap\left(\ww^{(i,k)}
      +
      \alpha \abs{\Delta}\right)
    \\
    \leq
    & \gammap\left(\ww \right) +  \alpha \abs{\Delta^{\top}
      \gamma'\left(\ww \right)} + p^{2}2^{p-3} \alpha^{2} \sum_e
      \max\{m^{\nfrac{1}{p}} \tt, \abs{\ww_e}, \alpha \Delta_e\}^{p-2}
      \Delta_e^2, 
      \quad \text{by Lemma \ref{lem:FirstOrderAdditive}}\\
    \leq
    & \gammap\left(\ww \right) + p \alpha \gammap\left(\ww
      \right)^{\frac{p-1}{p}}
      + p \alpha m^{\frac{p-2}{2p}} \gammap\left(\ww \right)^{1/2} \\
    & + p^{2} 2^{p-3} \alpha^{2} \sum_e \left( \rr_e\Delta_{e}^{2} +
      \alpha^{p-2} \Delta_e^{p} \right) , \quad \text{by Lemma
      \ref{lem:Oracle}}\\
    \leq
    & \gammap\left(\ww \right) + p \alpha \gammap\left(\ww
      \right)^{\frac{p-1}{p}}
      + p \alpha m^{\frac{p-2}{2p}} \gammap\left(\ww \right)^{1/2} \\
    & + p^{2} 2^{p-3} \left(\alpha^{2}  m^{\frac{p-2}{p}} +
      \alpha^{2} \pnorm{\ww}^{p-2}  +  \sum_e  \alpha^{p}  \Delta_e^{p}
      \right), \quad \text{by Lemma \ref{lem:Oracle}}\\
    \leq
    & \gammap\left(\ww \right) + p \alpha \gammap\left(\ww
      \right)^{\frac{p-1}{p}}
      + p \alpha m^{\frac{p-2}{2p}} \gammap\left(\ww \right)^{1/2} \\
    & + p^{2} 2^{p-3} \left(\alpha^{2}  m^{\frac{p-2}{p}} +
      \alpha^{2} \pnorm{\ww}^{p-2}  +  \alpha m^{\frac{p-1}{p}} 
      \right), \quad \text{by Assumption~\ref{enu:pPowerStep} of this
      Lemma}\\
\intertext{Using $\pnorm{\ww} \le \gammap(\ww)$}
    \leq
    & \gammap\left(\ww \right) + p \alpha \gammap\left(\ww
      \right)^{\frac{p-1}{p}}
      + p \alpha m^{\frac{p-2}{2p}} \gammap\left(\ww \right)^{1/2} \\
    & + p^{2} 2^{p-3} \left(\alpha^{2}  m^{\frac{p-2}{p}} +
      \alpha^{2} \gammap(\ww)^{\frac{p-2}{p}}  +  \alpha m^{\frac{p-1}{p}} 
      \right)\\
   \intertext{ Recall $\gammap(\ww) = \Phi(\ww).$ Letting $z$ denote
    $\max\{\Phi(\ww),m\}^{{\nfrac{1}{p}}}$, we have,}
    \leq
    & z^{p} + p \alpha z^{(p-1)}
      + p \alpha z^{(p-1)} \\
    & + p^{2} 2^{p-3} \left(\alpha^{2}  z^{(p-2)} +
      \alpha^{2} z^{p-2}  +  \alpha z^{(p-1)} 
      \right)\\
    \leq
    & z^{p} + p^{2} 2^{p-2} \alpha z^{(p-1)}\\
    & + p^{2} 2^{p-2} \alpha^{2}  z^{(p-2)}\\
  \leq
  & (z + p^{2}2^{p}  \alpha)^{p}.
  \end{align*}

  From the inductive assumption, we have
  \begin{align*}
  z & = \max\left\{\Phi,m\right\}^{\nfrac{1}{p}} \le
    \max\left\{\left(
    {p^{2}2^{p}\alpha i} + m^{\nfrac{1}{p}} \right)^{p}
    \exp\left(O_p(1) \frac{k}{\rho^2 m^{2/p} \beta^{-\frac{2}{p-2}}} \right), 
      m\right\}^{{\nfrac{1}{p}}}\\
    & =  \left(
    {p^{2}2^{p}\alpha k_{1}} + m^{\nfrac{1}{p}} \right)
    \left(\exp\left(O_p(1) \frac{k}{\rho^2 m^{2/p} \beta^{-\frac{2}{p-2}}} \right)\right)^{1/p}.
   \end{align*}
    Thus,
    \[
      \Phi(i+1,k)  \le (z + p^{2}2^{p} \alpha)^{p}
      \le \left({p^{2}2^{p}\alpha (i+1)} +
        m^{\nfrac{1}{p}}\right)^{p}
    \exp\left(O_p(1) \frac{k}{\rho^2 m^{2/p} \beta^{-\frac{2}{p-2}}} \right)
    \]
    proving
    the inductive claim.

 \paragraph{%
Width Reduction Step.}
To analyze a width-reduction step, we first observe that, by
Lemma~\ref{lem:Oracle} and the induction hypothesis, which ensures
$\norm{\ww^{(i,k)}}_p^{p} \leq \Phi \leq O_p(1) m$,
and hence $\sum_{e} \rr_e f_e^2 \leq O_p(1) m^{(p-2)/p}$
so we have
\[
\sum_{e \in H} \rr_e \leq \rho^{-2} \sum_{e \in H} \rr_e f_e^2 \leq
\rho^{-2} \sum_{e} \rr_e f_e^2 \leq \rho^{-2} O_p(1) m^{(p-2)/p}.
\]
  Thus, when the next step is a width-reduction step, we have,
  \begin{align*}
    \Phi(i,k+1)
    & \le \Phi + O_p(1) \sum_{e \in H} \rr_e^{\frac{p}{p-2}}
\\ 
    & \le \Phi + O_p(1) \left( \sum_{e \in H} \rr_e \right)\left( \max_{e \in
      H} \rr_e \right) ^{\frac{p}{p-2}-1} 
\\ 
    & \le \Phi + O_p(1) \left( \rho^{-2} m^{\frac{p-2}{p}} \right)
      \beta^{\frac{2}{p-2}} 
\intertext{ Letting $z$ denote
    $\max\{\Phi(\ww),m\}^{{\nfrac{1}{p}}}$, we have,}
    & \le z \left( 1+ O_p(1) \left( \rho^{-2} m^{-\frac{2}{p}} \right)
      \beta^{\frac{2}{p-2}} \right).
 \end{align*}
 
    Thus,
    \[
      \Phi(i,k+1)  
\le 
z \left( 1+ O_p(1) \left( \rho^{-2} m^{-\frac{2}{p}} \right)
      \beta^{\frac{2}{p-2}} \right)
      \le \left({p^{2}2^{p}\alpha i} +
        m^{\nfrac{1}{p}}\right)^{p}
    \exp\left(O_p(1) \frac{k+1}{\rho^2 m^{2/p} \beta^{-\frac{2}{p-2}}} \right)
    \]
    proving
    the inductive claim.

\end{proof}

\newcommand{\Solver}{\textsc{Solver}}
\newcommand{\EnhancedSolver}{\textsc{EnhancedSolver}}
\newcommand{\ggtilde}{\widetilde{\gg}}

\section{Solving L2 problems}
\label{sec:L2Solver}

\begin{lemma}
  Given an algorithm {\Solver} for solving
  $\BB^{\top} \RR^{-1} \BB \xx = \dd,$ for a $m \times n$-fixed matrix
  $\BB,$ a fixed positive diagonal matrix $\RR > 0$ and an arbitrary
  vector $\dd,$ there is an algorithm {\EnhancedSolver} that can solve
  \begin{align}
    \begin{aligned}
      \min_{\ff} \quad & \frac{1}{2} \ff^{\top} \RR \ff \\
      \text{s.t.} \quad & \BB^{\top}\ff = 0 \\
      & \gg^{\top} \ff = z
    \end{aligned}
  \end{align}
  with one call to {\Solver}, two multiplications of $\BB$ with a
  vector, and an additional $O(m+n)$ time, if we assume
  \[ \gg^{\top} \RR^{-1} \BB \left( \BB^{\top} \RR^{-1} \BB
    \right)^{-1} \BB^{\top} \RR^{-1} \gg < \gg^{\top} \RR^{-1} \gg.\]
\end{lemma}
\begin{proof}
  Introducing the Lagrangian multipliers $\vv, a$ respectively for the constraint
  $\BB^{\top} \ff= 0,$ and $\gg^{\top} \ff = z,$ we can write the
  Lagrangian as
  \begin{align*}
     \frac{1}{2} \ff^{\top} \RR \ff - \vv^{\top} \BB^{\top}\ff - a
    \left( \gg^{\top} \ff - z \right).
  \end{align*}
  Now, optimizing the Lagrangian with respect to an unconstrained $\ff$, allows us
  to write
  \[\ff = \RR^{-1} \left( \BB \vv + a\gg \right).\]
  Plugging this back, we can simplify our Lagrangian as
  \[-\frac{1}{2} \left( \BB \vv + a \gg \right)^{\top} \RR^{-1}
    \left( \BB \vv + a \gg \right) + az.\]
  Optimizing with respect to $a,$ gives us,
  \[a  = \frac{z-\gg^{\top} \RR^{-1} \BB \vv}{\gg^{\top} \RR^{-1}
      \gg}.\]
  Plugging this back, gives the Lagrangian as
  \[-\frac{1}{2} \vv^{\top} \BB^{\top}\RR^{-1} \BB \vv + \frac{\left(
        z-\gg^{\top}\RR^{-1}\BB \vv \right)^{2}}{2 \gg^{\top} \RR^{-1}
      \gg}.\]
  We let $\ggtilde$ denote the vector $\BB^{\top} \RR^{-1} \gg$ and
  $\MM$ denote the matrix $ \BB^{\top} \RR^{-1} \BB.$ Thus, the
  Lagrangian can be written as,
  \[-\frac{1}{2} \vv^{\top} \MM \vv + \frac{\left(
        z-\ggtilde^{\top} \vv \right)^{2}}{2 \gg^{\top} \RR^{-1}
      \gg}.\]
  This implies that the optimal $\vv$ is given by the equation
  \[\left( \MM - \frac{1}{\gg^{\top} \RR^{-1} \gg}
      \ggtilde \ggtilde^{\top} \right)\vv
    = - \frac{z \ggtilde}{\gg^{\top} \RR^{-1} \gg}.
  \]
  From the condition assumed on $\gg,$ we have
  $\ggtilde \MM^{-1} \ggtilde < \gg \RR^{-1} \gg.$
  Thus, we can solve this system using the Sherman-Morrisson formula
  as follows,
  \begin{align*}
    \vv & = -
          \left(
          \MM^{-1}
          + \frac
          {        \MM^{-1} \ggtilde \ggtilde^{\top} \MM^{-1}}
          { \gg^{\top} \RR^{-1} \gg -
          \ggtilde^{\top} \MM^{-1}  \ggtilde}
          \right)
          \frac
          {z \ggtilde}
          {  \gg^{\top} \RR^{-1} \gg} \\
        & = - \frac{z \MM^{-1} \ggtilde}{ \gg^{\top} \RR^{-1} \gg -
          \ggtilde^{\top} \MM^{-1}  \ggtilde}
  \end{align*}
  The algorithm {\EnhancedSolver} computes $\ggtilde,$ and then
  invokes {\Solver} to compute $\MM^{-1} \ggtilde.$ This allows us to
  compute $\vv$ is an additional $O(m+n)$ time. Finally, we can
  compute $\ff = \RR^{-1}(\BB \vv + a \gg)$ using another
  multiplication with $\BB$ and an additional $O(m+n)$ time.
\end{proof}

\section{General $\ell_2$ Resistance Monotonicity}
\label{sec:ckmstResIncreaseProof}

\ckmstResIncrease*

\begin{proof}
Recall
\begin{align*}
 \energy{\rr} = \min_{\Delta} \quad
  & \sum_e
  \rr_e \Delta_e^2\\
  \text{s.t. } \quad & \AA' \Delta = \cc 
\end{align*}
Letting $\RR$ denote the diagonal matrix with $\rr$ on its diagonal,
we can write the above as 
\begin{align}
\label{eq:l2primal}
 \energy{\rr} = \min_{\Delta} \quad
  & \Delta^{\top} \RR \Delta \\\nonumber
  \text{s.t. } \quad & \AA' \Delta = \cc 
\end{align}
Using Lagrangian duality, and noting that strong duality holds, we can write this as 
\begin{align*}
 \energy{\rr} = \min_{\Delta} \max_{\yy} & \quad
 \Delta^{\top} \RR \Delta + 2\yy^{\top} (\cc-\AA' \Delta )  \\
= \max_{\yy} \min_{\Delta}  &\quad  2\yy^{\top} \cc + 
  \Delta^{\top} \RR \Delta - 2\yy^{\top}\AA' \Delta
\end{align*}
The minimizing $\Delta$ can be found by setting the gradient w.r.t. to
this variable to zero.
This gives $2 \RR \Delta - 2\yy^{\top}\AA'=0$, so that 
$\Delta = \RR^{-1}(\AA')^{\top}\yy$.
Plugging in this choice of $\Delta$, we arrive at the dual program
\begin{align}
\label{eq:l2dual1}
 \energy{\rr} = \max_{\yy} & \quad
                             2 \cc^{\top} \yy  - \yy^{\top} \AA'\RR^{-1}(\AA')^{\top}\yy
\end{align}
Crucially, strong duality also implies that if $\Delta^*$ is an
optimal solution of the primal program~\eqref{eq:l2primal}, and
$\yy^*$ is an optimal solution to the dual then 
\[
\min_{\Delta}  \quad  2(\yy^*)^{\top} \cc + 
  \Delta^{\top} \RR \Delta - 2(\yy^*)^{\top}\AA' \Delta
\]
is optimized at
  $\Delta = \Delta^*$.
This in turn implies the gradient w.r.t. $\Delta$ at $\Delta =
\Delta^*$ is zero, so that $\Delta^* = \RR^{-1}(\AA')^{\top}\yy^*$.
Let $\aa_e$ be the $e$th row of $\AA'$.
Then the previous equation tells us that
$\Delta^*_e = \frac{1}{\rr_e} \aa_e^{\top} \yy^*$.
This implies that 
\begin{equation}
\forall e. \rr_e (\Delta^*_e)^2 = \frac{1}{\rr_e} (\aa_e^{\top}
\yy^*)^2.\label{eq:edgewisePrimalDual}
\end{equation}

Consider another program, essentially the same as~\eqref{eq:l2dual1}, but with additional
scalar valued variable $\theta\in\rea$  introduced.
\begin{align}
\label{eq:l2dual2}
 \energy{\rr} = \max_{\zz,\theta} & \quad
                             \theta \cdot 2 \cc^{\top} \zz  - \theta^2
                                    \cdot \zz^{\top} \AA' \RR^{-1} (\AA')^{\top}\zz
\end{align}
The two programs~\eqref{eq:l2dual1} and~\eqref{eq:l2dual2} have the same value, since for any $\yy$,
the assignment $(\zz,\theta) = (\yy,1)$ ensures both objectives take
the same value, and conversely for any $(\zz,\theta)$, the assignment
$\yy = \theta \zz$  ensures both programs take the same value.

We see that $\zz = \yy^*,$ and $\theta = 1$ is an optimal solution to
~\eqref{eq:l2dual2}. Hence 
\[
\left[
\frac{d}{d \theta}
\left(
\theta \cdot 2 \cc^{\top} \yy^*  - \theta^2
                                    \cdot (\yy^*)^{\top} \AA'\RR^{-1} (\AA')^{\top}\yy^*
\right)
\right]_{\theta = 1} = 0
\]
Consequently, $\cc^{\top} \yy^* = (\yy^*)^{\top}\AA'\RR^{-1} (\AA')^{\top}\yy^*$.
Hence $\cc^{\top} \yy^* = \energy{\rr}$.
Again, by a scaling argument, this implies that 
\begin{align*}
2-\frac{1}{\energy{\rr} } = 
&  \max_{\yy}  \quad
2 \cc^{\top} \yy  - \yy^{\top} \AA'\RR^{-1}
                                    (\AA')^{\top}\yy
\\
& \text{s.t. } \quad  \cc^{\top} \yy = 1
\end{align*}
So that 
\begin{align}
\label{eq:l2dualinverse}
\frac{1}{\energy{\rr} } = 
&  \min_{\yy}  \quad
\yy^{\top} \AA'\RR^{-1}
                                    (\AA')^{\top}\yy
\\\nonumber
& \text{s.t. } \quad  \cc^{\top} \yy = 1
\end{align}
Note that one optimal assignment for the
program~\eqref{eq:l2dualinverse} is  
$\yytil = \frac{\yy^*}{\energy{\rr}}$.
Also observe that if we consider the program \eqref{eq:l2dualinverse} with
$\rr'$ instead of $\rr$ as the resistances, then $\yytil
= \frac{\yy^*}{\energy{\rr}}$ is still a feasible solution.
Hence, using the observation~\eqref{eq:edgewisePrimalDual}, we get
\[
\frac{1}{\energy{\rr'} }
\leq \frac{1}{\energy{\rr}^2}
\left(\yy^*\right)^{\top} \AA'
\left(\RR'\right)^{-1}
\left(\AA'\right)^{\top} \yy^*
=
\frac{1}{\energy{\rr}^2}
\sum_{e}\frac{\rr_e}{\rr_e'} \rr_e(\Delta^*_e)^2
\]
Factoring out the $\Psi(\rr)$ term gives
\[
\frac{1}{\energy{\rr'} }
\leq
\frac{1}{\energy{\rr}}
\left( 1 - 
\frac{\sum_{e} \left(1 -  \frac{\rr_e}{\rr_e'} \right) \rr_e(\Delta^*_e)^2}
{\energy{\rr}} \right)
\leq
\exp\left( - 
\frac{\sum_{e} \left(1 -  \frac{\rr_e}{\rr_e'} \right) \rr_e(\Delta^*_e)^2}
{\energy{\rr}} \right).
\]
Now consider the term $1 - \frac{\rr_e}{\rr'_e}$:
if $\rr'_e \geq 2 \rr_e$, then it is at least $1/2$.
Otherwise, it can be rearranged to
\[
\frac{\rr'_e - \rr_e}{\rr'_e}
\geq
\frac{\rr'_e - \rr_e}{2\rr_e}.
\]
So in either case, we have
\[
\frac{1}{\energy{\rr'} }
\leq
\frac{1}{\energy{\rr}}
\exp\left( - 
\frac{\sum_{e} \min\left\{1, \frac{\rr'_e - \rr_e}{\rr_e} \right\}
\rr_e(\Delta^*_e)^2}
{2\energy{\rr}} \right),
\]
which upon rearranging gives the desired result.
\end{proof}

\end{document}